\documentclass[journal]{IEEEtran}
\usepackage{amsfonts}
\usepackage{amssymb}
\usepackage{amsmath}
\usepackage{algorithmic}
\usepackage{array}
\usepackage{soul}
\usepackage{array}
\usepackage{makecell}
\usepackage{enumitem}
\usepackage{multirow}
\usepackage{tabularx}
\usepackage{booktabs}
\usepackage[table]{xcolor}
\usepackage{textcomp}
\usepackage[ruled]{algorithm2e}
\usepackage{stfloats}
\usepackage{dsfont}
\usepackage{url}
\usepackage{xcolor}
\usepackage{subcaption}
\usepackage{verbatim}
\usepackage{graphicx}
\usepackage{amssymb}
\allowdisplaybreaks[4]
\usepackage{cite}
\usepackage{upgreek}
\usepackage{pifont}

\usepackage{amsthm}
\usepackage{amssymb}
\usepackage{makecell}
\usepackage{enumitem}

\newtheorem{definition}{Definition}
\theoremstyle{remark}
\theoremstyle{proposition}
\newtheorem{remark}{Remark}
\newtheorem{proposition}{Proposition}
\theoremstyle{Assumption}

\theoremstyle{Lemma}

\definecolor{red}{rgb}{1,0,0} 
\definecolor{green}{rgb}{0,0.7,0}  
\definecolor{myblue}{RGB}{199,217,236}
\definecolor{mygrey}{RGB}{236,236,237}

\ifodd 0
\newcommand{\rev}[1]{{\color{blue}#1}} 
\newcommand{\newrev}[1]{{\color{red}#1}} 
\else
\newcommand{\rev}[1]{#1}
\newcommand{\newrev}[1]{#1} 
\fi

\begin{document}

\title{RAISE: Optimizing RIS Placement to Maximize Task Throughput in Multi-Server Vehicular Edge Computing}


\author{Yanan Ma,~\IEEEmembership{Student Member, IEEE}, Zhengru Fang, Longzhi Yuan, Yiqin Deng,~\IEEEmembership{Member,~IEEE}, \\
Xianhao Chen,~\IEEEmembership{Member,~IEEE}, and Yuguang Fang,~\IEEEmembership{Fellow,~IEEE}
\thanks{Y. Ma, Z. Fang, L. Yuan, Y. Deng, and Y. Fang are with Hong Kong JC Lab of Smart City and the Department of Computer Science, City University of Hong Kong, Hong Kong, China. E-mail: \{yananma8-c, zhefang4-c\}@my.cityu.edu.hk, \{longyuan, yiqideng, my.fang\}@cityu.edu.hk.}

\thanks{X. Chen is with the Department of Electrical and Electronic Engineering, University of Hong Kong, Hong Kong, China. E-mail: xchen@eee.hku.hk.}
}



\maketitle

\begin{abstract}
Given the limited computing capabilities on autonomous vehicles, onboard processing of large volumes of latency-sensitive tasks presents significant challenges. While vehicular edge computing (VEC) has emerged as a solution, offloading data-intensive tasks to roadside servers or other vehicles is hindered by large obstacles like trucks/buses and the surge in service demands during rush hours. To address these challenges, Reconfigurable Intelligent Surface (RIS) can be leveraged to mitigate interference from ground signals and reach more edge servers by elevating RIS adaptively. To this end, we propose RAISE, an optimization framework for \ul{R}IS pl\ul{a}cement in mult\ul{i}-\ul{s}erver V\ul{E}C systems. Specifically, RAISE optimizes RIS altitude and tilt angle together with the optimal task assignment to maximize task throughput under deadline constraints. To find a solution, a two-layer optimization approach is proposed, where the inner layer exploits the unimodularity of the task assignment problem to derive the efficient optimal strategy while the outer layer develops a near-optimal hill climbing (HC) algorithm for RIS placement with low complexity. Extensive experiments demonstrate that the proposed RAISE framework consistently outperforms existing benchmarks.

\end{abstract}


\section{Introduction}
\IEEEPARstart{C}{onnected} and autonomous driving (CAD) is expected to profoundly revolutionize the future of transportation by delivering safer, more efficient, and more comfortable driving experiences \cite{XianhaoVaaS}. To achieve this goal, autonomous vehicles (simply vehicles at later development) must gather huge amounts of real-time data from various onboard sensors and generate a vast number of computing tasks for smart driving and entertainment \cite{Dataoffloading}, \cite{PACP}, \cite{Zhengru}. However, processing such massive multi-modal tasks within sub-seconds, as demanded by time-sensitive vehicular applications, presents a major hurdle that precludes the widespread deployment of CAD. Given the substantial deployment costs of computing power, a cost-effective solution is to offload some computing tasks from vehicles to multi-access edge computing (MEC) systems at roadside or more powerful vehicles located nearby~\cite{AD}. With the assistance of MEC, the basic safety-critical tasks may still be processed onboard to ensure reliability and safety, whereas advanced driving tasks, such as trajectory planning, human-vehicle interactions, augmented reality (AR) navigation, and real-time interactive gaming, can be offloaded to MEC systems with more powerful computing/artificial intelligence (AI) capabilities. This strategy represents an emerging paradigm known as vehicular edge computing (VEC) \cite{VEC}.

VEC systems often handle \textit{data-intensive} and \textit{delay-sensitive} tasks, which require vehicles to establish fast communications with MEC servers that provide powerful computing capabilities~\cite{AD}. Unfortunately, fulfilling these task requirements is fraught with significant challenges. On the one hand, the absence of line-of-sight (LoS) channels results in low data rates and significant transmission latency between vehicles and roadside VEC servers. The relatively \textit{low altitude} of VEC servers, e.g., co-locating with base stations (BSs), access points (APs), and roadside units (RSUs), makes communication links vulnerable to physical blockages, including green belts, trucks, and buses. The situation is further exacerbated by the high mobility of vehicles, making vehicle-to-server connections highly unstable~\cite{ShenJSAC}. On the other hand, the high density of vehicles in populated areas, particularly during peak hours, causes a surge in data transmission and computation workload. The intensive communication-computation demands can easily overwhelm resource-limited roadside VEC servers, resulting in prolonged transmission and service waiting time.

\begin{figure}[t]
    \centering
    \begin{subfigure}[b]{0.24\textwidth}
    \includegraphics[width=\textwidth]{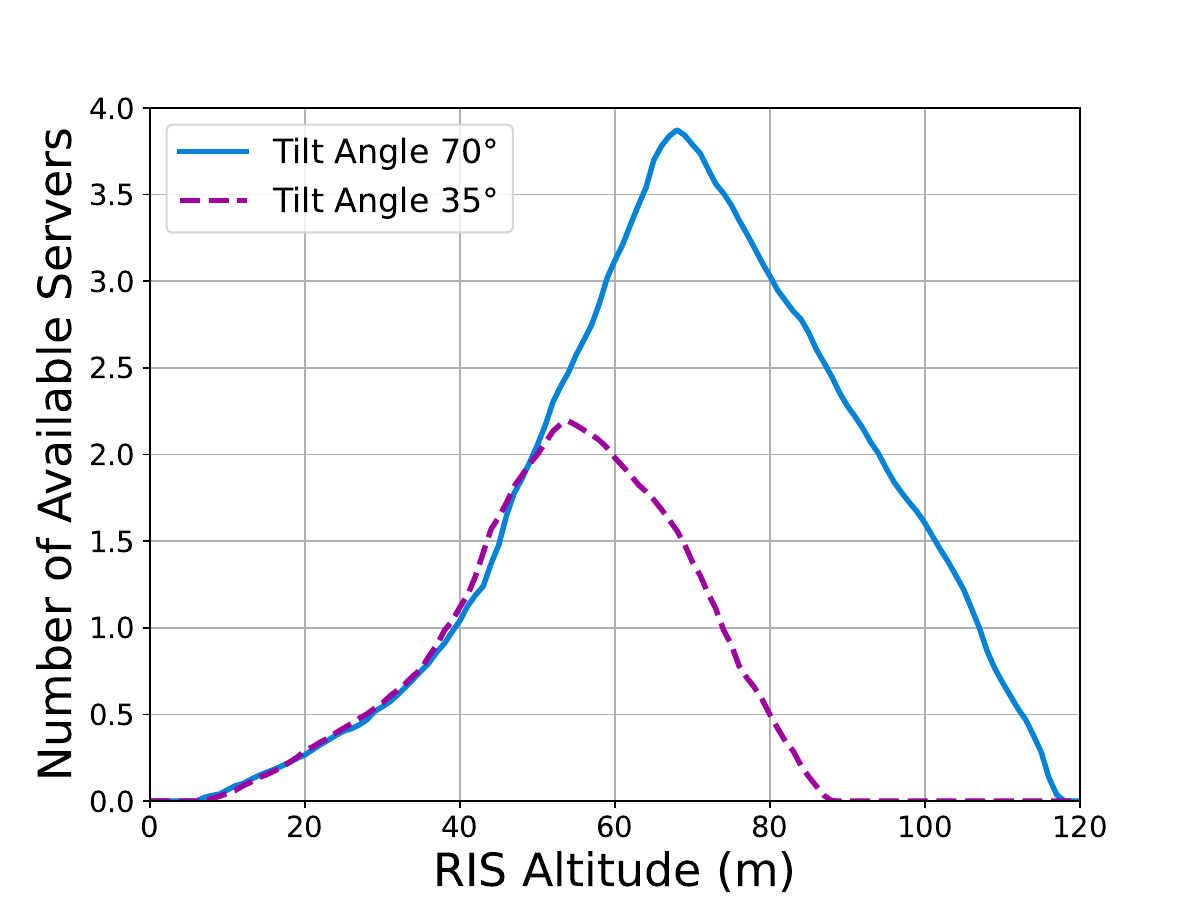}  
    \caption{Number of available servers versus RIS altitude.}
    \label{fig: intro_height}
\end{subfigure}
    \hfill
\begin{subfigure}[b]{0.24\textwidth}
    \includegraphics[width=\textwidth]{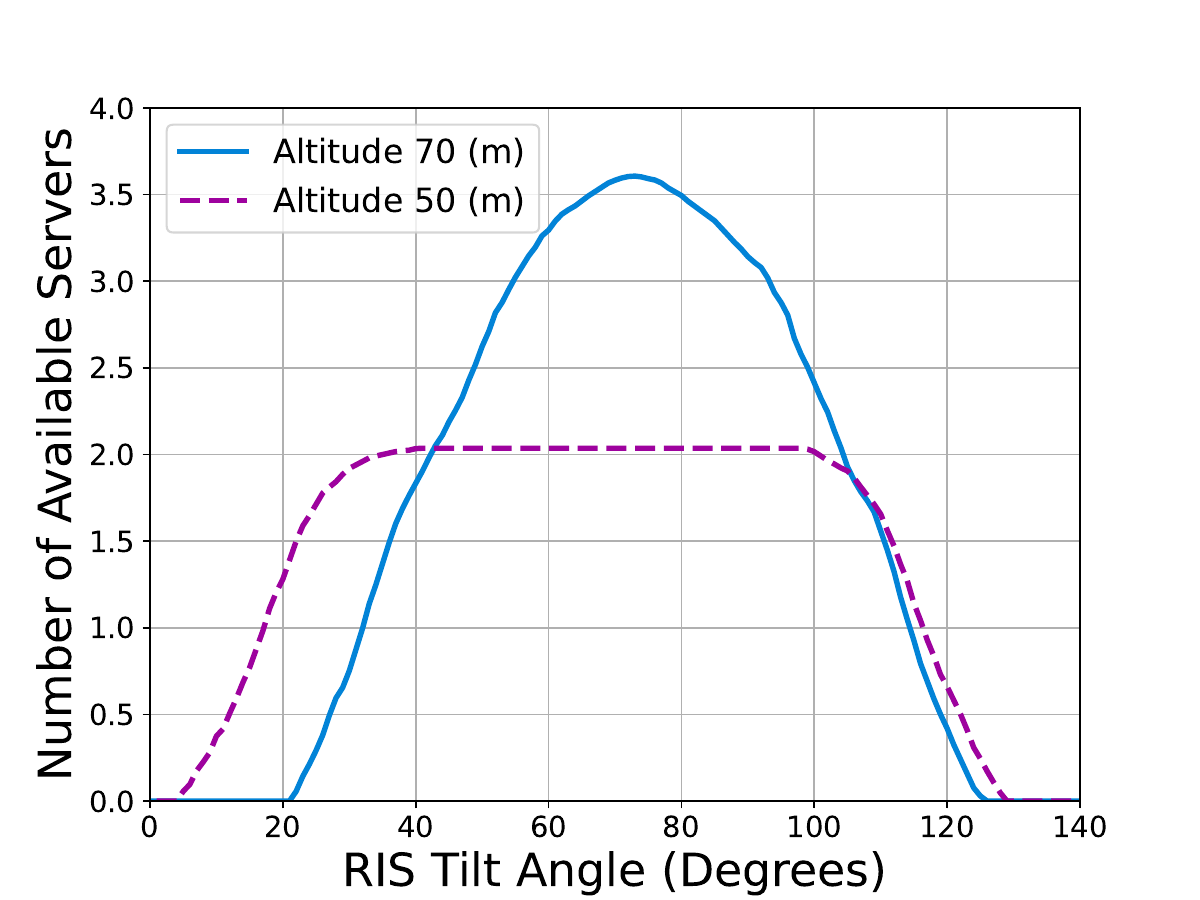}  
    \caption{Number of available servers versus RIS tilt angle.}
    \label{fig: intro_angle}
\end{subfigure}
    \caption{The impact of RIS placement on the average number of available servers for a vehicle. An available server refers to a server that can fulfill the offloaded tasks from the vehicle with a deadline requirement, which is set to $100$ ms.}
    \label{fig: intro}
\end{figure}



In recent years, reconfigurable intelligent surface (RIS), an intelligent surface capable of augmenting transmission environments, has attracted significant attention~\cite{Zhanglan, Qingqing, yanan, yananEE, ZengARIS, Tang, MM, Manifold}. Strategic deployment of RISs in VEC systems can effectively address the aforementioned challenges. First of all, elevated RISs, such as those mounted on lampposts or carried by unmanned aerial vehicles (UAVs), can establish relaying LoS links by reflecting signals from vehicles toward roadside servers \cite{ZengARIS}. Given the fact that direct links between vehicles and roadside servers are frequently obstructed in rapidly changing vehicular environments, these vehicle-RIS-server links can ensure high-speed and stable communications. 
Second, elevated RISs also facilitate vehicles to beam signals upwards, mitigating wireless interference caused by interlaced and congested electromagnetic (EM) signals close to the ground, thereby improving network-wide spectrum efficiency. Third, elevated RIS can help vehicles access more roadside servers by reflecting tasks to less loaded but potentially distant edge servers, thus achieving a better balance of communication and computation workloads in VEC systems \cite{Fang}. 
Fourth, while elevated relays can serve a similar function, RISs merely reflect signals without introducing significant processing delays or half-duplex transmission delays \cite{yiqin}, making them more suitable for time-critical VEC applications. \rev{Finally, RIS deployment is cost-effective due to its lightweight design, ease of installation on walls or racks, and energy efficiency. In particular, RIS operates passively, requiring minimal energy for the phase shift controller and control signaling exchanges without the need to perform user data transmissions~\cite{Tang}.}

Based on the aforementioned observations, in this paper, we investigate the joint problem of RIS positioning and task offloading in multi-server VEC systems. As depicted in Fig. \ref{fig: intro}, we consider a vehicle with a task that has a deadline requirement. The number of available servers (defined as the servers that can fulfill the task) initially increases with the RIS altitude, as a higher altitude provides a broader field of view. However, as altitude increases, communication link quality degrades due to path attenuation, resulting in prolonged communication delay. Thus, beyond a certain altitude, further increasing the altitude leads to a decline in the number of available servers that meet the overall task completion latency due to the prolonged communication delay. Besides, the RIS tilt angle requires careful adjustment because it should be precisely oriented toward the target areas to maximize connectivity. These observations emphasize the importance of carefully optimizing the RIS positioning, i.e., both the RIS altitude and tilt angle, to enhance the overall task computing performance and offer valuable insights into designing efficient and effective heuristic algorithms.

Our objective in this paper is to optimize both the positioning of a RIS and task assignments to maximize \textit{task throughput}, defined as the number of offloaded computing tasks that can be completed within their deadline constraints. A fundamental tradeoff in RIS placement involves balancing two conflicting factors of increasing the RIS altitude: \textit{a higher RIS enables vehicles to connect to more VEC servers by circumventing ground obstacles (hence occlusions), but it also results in higher path loss (low communication quality).} Moreover, RIS positioning is complicated by the task offloading strategy, as these two coupled decisions jointly determine the \textit{optimal performance} of this multi-server system. Since the joint optimization problem is highly intractable, \newrev{we search for the RIS placement decisions in the outer layer and optimize task offloading in the inner layer. In the outer layer, by analyzing the solution space, we developed a hill-climbing algorithm, yielding near-optimal solutions with low complexity. This method is based on our experimental insights and intuition: the task throughput, in terms of RIS placement with optimal task offloading decisions, tends to exhibit a hill-like shape. For task offloading in the inner layer, we demonstrate the unimodularity of the linear integer programming problem, showing that the task offloading problem is equivalent to its relaxed linear programming counterpart, which can be solved very efficiently. The efficiency of the inner layer makes it feasible to solve it in each iteration of the hill-climbing algorithm. The aforementioned design enables us to obtain near-optimal solutions of the very challenging optimization problem. The main contributions of this paper are summarized as follows.}
\begin{itemize}
    \item By considering a multi-server VEC system, we formulate a network-level optimization problem of task offloading and RIS placement, i.e., the altitude and tilt angle, to maximize the task throughput under probabilistic (soft) task deadline constraints. The formulation takes radiation characteristics of the RIS, probabilistic channel models, and mobility of vehicles into account.  

    \item We develop a two-layer optimization approach to find effective solution. In the task offloading problem, by demonstrating the unimodularity of the linear integer programming problem, we show that the problem is equivalent to the relaxed linear programming counterpart, leading to the optimal offloading solutions.
    



\item We employ a grid search procedure to optimally determine the RIS placement for an appropriate problem scale. We further design a near-optimal HC algorithm to optimize the RIS placement more efficiently.



    \item Finally, we conduct extensive simulation results to demonstrate that our low-complexity algorithm significantly outperforms other benchmarks. We also offer valuable insights on how to deploy RIS in VEC systems.

\end{itemize}

The rest of this paper is organized as follows. Section \ref{sec: related work} reviews related works on RIS-assisted MEC systems and RIS placement problems. Section \ref{sec:system} presents the system model. We formulate the RIS placement problem for task throughput maximization and provide the efficient optimization methods in Section \ref{sec: problem and alg}. In Section \ref{sec: simulation}, we provide extensive simulation results and discussions to demonstrate the effectiveness of our framework. Finally, we conclude this paper in Section \ref{sec:conclusion}.


\section{Related Work}
\label{sec: related work}

By harnessing the benefits of RISs, RIS-assisted MEC systems have garnered widespread attention. For instance, He \textit{et al.} propose balancing computing workloads by leveraging RIS to redirect computing tasks to less overloaded edge servers~\cite{He}. 
\rev{A joint optimization problem involving user association, passive beamforming at the RIS, receive beamforming at BSs, and computing resource allocation at servers is formulated to maximize task completion rates.}
Chu \textit{et al.} maximize the volume of accomplished tasks in RIS-assisted MEC systems by jointly optimizing computing resource, transmit power, time allocation, and RIS phase shifts~\cite{Chu}. 
Bai \textit{et al.} minimize the weighted sum latency in RIS-assisted MEC systems through joint optimization of offloading data size, computing resource, and both active and passive beamforming~\cite{Bai}. 
Yu \textit{et al.} address the joint optimization problem of hybrid beamforming at BS, passive beamforming at RIS, and computing resource for MEC systems to maximize computation efficiency \cite{Yu}. 
Shang \textit{et al.} introduce aerial RIS (ARIS) mounted on a UAV into MEC \rev{to enable three-dimensional signal reflections for uplink computation offloading. They focus on optimizing communication and computing resource allocation, ARIS trajectory, and the amplitudes and phase shifts of ARIS} \cite{poor}. 
In \cite{URIS}, Zhai \textit{et al.} propose a UAV-mounted RIS-assisted MEC system to enhance energy efficiency by jointly optimizing the passive beamforming at RIS, UAV trajectories, and MEC computing resource. However, the altitude and tilt angle of the RIS have not been optimized. Instead of solely using a RIS to enhance link quality, Zhang \textit{et al.} propose to use a novel Reconfigurable Intelligent Computational Surface (RICS) to improve computing by dynamically adjusting the incident signal’s amplitude, thereby mitigating interference at the receiver in vehicle-to-vehicle communication \cite{RICS}, which aims to improve autonomous driving by optimizing task offloading ratios, spectrum sharing strategies, and the RICS reflection and refraction matrices. In addition, several studies have explored RIS-assisted secure computation offloading and resource management for MEC systems \cite{Mao, Michailidis}.


\begin{table}[!t]
\centering
\caption{\rev{Summary of related works on RIS-assisted systems}}
\label{table_comp}
\renewcommand{\arraystretch}{1.4}
\setlength{\tabcolsep}{2mm}
\rev{
\begin{tabular}{|c|c|c|c|c|}
\hline
\makecell[c]{\textbf{Ref.}} & 
\makecell[c]{\textbf{MEC} \\ \textbf{Systems}}& 
\makecell[c]{\textbf{Multiple} \\ \textbf{Servers}}& 
\makecell[c]{\textbf{RIS Placement} \\ \textbf{(Altitude)}}& 
\makecell[c]{\textbf{RIS Placement} \\ \textbf{(Tilt Angle)}}   
  \\ \hline
\cite{He} 
    & {\ding{52}}  & {\ding{52}}  & {\ding{55}}  & {\ding{55}}   \\ \hline
\cite{Chu} 
    & {\ding{52}}  & {\ding{55}}  & {\ding{55}}  & {\ding{55}}   \\ \hline 
\cite{Bai} 
    & {\ding{52}}  & {\ding{55}}  & {\ding{55}} & {\ding{55}}   \\ \hline
\cite{Yu} 
    & {\ding{52}}  & {\ding{55}}  & {\ding{55}} & {\ding{55}}   \\ \hline
\cite{poor} 
    & {\ding{52}}  & {\ding{55}}  & {\ding{55}} & {\ding{55}}   \\ \hline
\cite{URIS} 
    & {\ding{52}}   & {\ding{55}}  & {\ding{55}}   & {\ding{55}} \\ \hline
\cite{Boya} 

    & {\ding{55}}    & {\ding{55}}  & {\ding{55}}     & {\ding{52}}  \\ \hline
\cite{YuyinTVT}
    & {\ding{55}}    & {\ding{55}}  & {\ding{55}}     & {\ding{52}}  \\ \hline
\cite{ChengTWC} 
    & {\ding{55}}  & {\ding{55}}   &  {\ding{52}} & {\ding{52}} \\ \hline
\cite{Xiaowen} 
    & {\ding{55}}  & {\ding{55}}   &  {\ding{52}} & {\ding{52}} \\ \hline
Ours 
    & {\ding{52}}  & {\ding{52}}  & {\ding{52}}   & {\ding{52}}    
     \\ \hline
\end{tabular}
}
\end{table}

Although RIS placement plays a pivotal role, very limited work has been devoted to optimizing RIS placement in MEC systems, with most of them still focusing on purely enhancing communication rather than MEC computing~\cite{He}. Zeng \textit{et al.} optimize RIS orientation and horizontal distance by addressing the RIS placement optimization problem for coverage maximization~\cite{Boya}. A multi-RIS location optimization scheme has been developed in \cite{YuyinTVT} to enhance coverage while reducing network costs. To take advantage of both relay and RIS, Bie \textit{et al.} derive a closed-form upper bound of the achievable rate and the optimal RIS deployment \cite{Bie}. Cheng \textit{et al.} further investigate the RIS placement problem by considering its radiation characteristics, showing that rotating RIS is more effective than moving it across a large area \cite{ChengTWC}. Tian \textit{et al.} determine the optimal RIS location, altitude, and tilt angle through the numerical search for an mmWave vehicular communication system by considering near-field beamforming and the distribution of users and obstacles \cite{Xiaowen}.


However, the aforementioned works cannot be applied to our RIS placement problem in either MEC or VEC systems for the following reasons. First, our RIS placement aims to maximize the number of completed computing tasks (called task throughput) within their specified deadline requirements under computing resource availability, which is fundamentally different from traditional rate-centric RIS placement designed to maximize the total data rate or coverage. This distinction is particularly important in multi-server systems, where the joint design of task offloading and placement must be addressed to achieve high task throughput. Second, the aforementioned works fail to consider the dynamic mobility of vehicles, which also impacts the optimal RIS position. To bridge these gaps, this paper investigates an optimal placement of RIS to balance communication and computing workloads in VEC systems. \rev{To compare our work and related works, we provide a summary table in Table \ref{table_comp}.}


\section{System Model}
\label{sec:system}
In this section, we describe our RIS-assisted VEC system, vehicle mobility model, probabilistic communication model, and task latency model by considering the effect of RIS placement.

\subsection{RIS-assisted VEC Systems}

\begin{figure}
    \centering
    \includegraphics[width=1\linewidth]{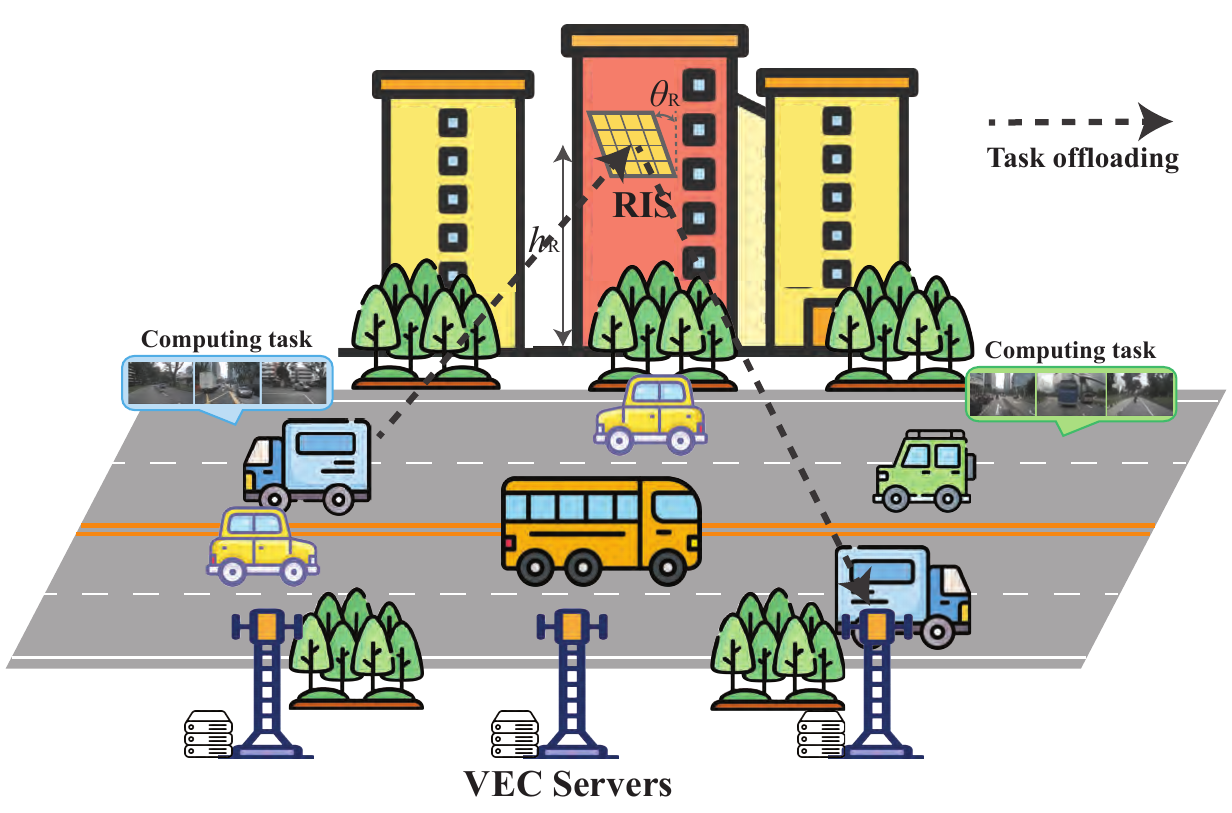}
    \caption{An elevated RIS-assisted multi-server VEC system. Multiple vehicles generate computing tasks while moving through the area. The RIS is elevated and tilted down to facilitate task offloading by enhancing the reachability to more computing servers.}
    \label{fig: system}
\end{figure}

We consider a RIS-assisted VEC system, as illustrated in Fig. \ref{fig: system}, where a RIS can be either installed on building facades or lampposts to facilitate task offloading for vehicles. \rev{To enable network control, RIS can be equipped with a transceiver for channel estimation and control signaling exchange, thereby facilitating real-time adaptation of its adjustable variables like phase shifts~\cite{Qingqing, Tang}. To manage the system, we assume there is a central controller, say, co-located with a base station for network coordination. Since the 5G/6G network can adopt a separate control and data plane following the principle of control/user plane decoupling, the central controller can collect network information and deliver control commands through dedicated control channels to inform the policies to the RIS controller.}

The position of the RIS is given by $(x_{\mathrm{R}}, y_{\mathrm{R}}, h_{\mathrm{R}})$, where $x_{\mathrm{R}}$ and $y_{\mathrm{R}}$ on the horizon are both pre-determined, while $h_{\mathrm{R}}$, representing the altitude  or height, should be optimized. The RIS is composed of $M_r \times M_l$ reconfigurable elements, where $M_r$ and $M_l$ denote the numbers of rows and columns of regularly spaced elements. Moreover, the RIS is tilted downwards with a tilt angle $\theta_{\mathrm{R}}$. Since the placement of the RIS has a long-term impact on system performance, we treat it as a strategic, long-term decision. 
Initially, we consider a single instance for the sake of clarity. Then, we will address the long-term optimization of RIS placement across multiple instances. 




Let $\mathcal{K}=\{1, 2, \dots, K\}$ represent the set of vehicles. The initial location of vehicle $k$ is $(x_k[1], y_k[1], z_k)$, which can be determined when it submits its task offloading request. The available VEC servers are denoted by $\mathcal{S}=\{1, 2, \dots, S\}$, with the location of server $s$ being $(x_s, y_s, z_s)$. Depending on task computing requirements, RIS altitude, and tilt angle, not all VEC servers can be used for task computing, and thus, $S$ can be considered to be the maximum number of VEC servers that can be put into use at the point of interest (e.g., the spot that the RIS is located at). Each vehicle generates one indivisible computing task.
Let $\lambda_{k,s}$ indicate the association between vehicle $k$ and VEC server $s$, where $\lambda_{k,s} = 1$ represents that vehicle $k$ offloads its computing task to server $s$, and $\lambda_{k,s} = 0$ otherwise. Assuming each task is either executed locally or offloaded to at most one VEC server, we have
\begin{equation}
    \sum_{s=1}^S \lambda_{k,s} \leq 1, ~\forall k \in \mathcal{K}.
\end{equation}

We characterize the task of vehicle $k$ using a tuple $\mathcal{T}_k(D_k, F_k, T^{th}_k)$, where $D_k$ is the input data size (in bits), $F_k$ is the computing workload (in floating point operations or FLOPs) per bit, and $T^{th}_k$ specifies the maximum tolerable delay (in seconds). These parameters can be determined by monitoring task execution and active reporting from vehicles when submitting their offloading requests \cite{3tuple}. To enable multiple access, Orthogonal Frequency Division Multiplexing Access (OFDMA) is adopted, with a fixed channel allocated to each vehicle~\cite{MECYuanwei}\footnote{This paper focuses on RIS placement in VEC systems. Spectrum resource allocation is beyond the scope of this paper, which can be explored in future studies.}.

\subsection{Vehicle Mobility Model}

Due to the high mobility of vehicles, their locations continuously change upon offloading computing tasks to VEC servers, causing data rate changes. To facilitate the analysis, we discretize the region into grids to represent vehicles' locations \rev{\cite{discrete}}, as shown in Fig. \ref{fig: mobility}, assuming that each vehicle's data rate remains unchanged within a grid but potentially varies across different grids. Supposing that each vehicle moves at a constant speed~\cite{speed1}, \cite{speed2}, the sojourn time $t_k$ for vehicle $k$ in a given grid can be calculated as
\begin{equation}
    t_k = \frac{\Delta_d}{v_k},
\end{equation}
where $\Delta_d$ represents the length of the grid and $v_k$ is the velocity of vehicle $k$\footnote{\rev{Here, we assume a constant velocity in this analysis for simplicity. However, the proposed framework can be easily generalized to time-varying velocity scenarios.}}. This approach can be easily extended to other mobility models, such as the widely used Intelligent Driver Model (IDM)~\cite{mobilitymodel}, the Gauss-Markov mobility model~\cite{mobility}, and the Krauss model \cite{Krauss}. Simulation results demonstrate that the proposed algorithm remains effective in practical vehicular flow scenarios.

\begin{figure}
    \centering
    \includegraphics[width=0.9\linewidth]{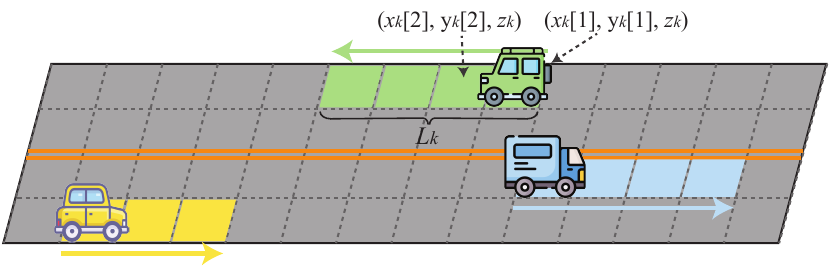}
    \caption{Illustration of vehicular mobility. The region is divided into multiple grids to represent vehicles' locations.}
    \label{fig: mobility}
\end{figure}

\subsection{Communication Model}
The initial position of vehicle $k$ is at $(x_k[1], y_k[1], z_k)$, and as it moves, its subsequent position in the $l$-th grid along its trajectory is denoted by $(x_k[l], y_k[l], z_k)$, as illustrated in Fig. \ref{fig: mobility}. Given the scattered road environments, we assume that both the transmitting and receiving beams of vehicles and servers are aligned with the RIS. In other words, vehicles offload their computing tasks through the strategically placed RIS via the cascaded vehicle-RIS-server communication channel.

We adopt a statistical model for the vehicle-RIS and RIS-server channels. Since shadowing and scattering are prevalent due to dense buildings in urban environments, we define a binary state \( a_{\mathrm{R}, k}[l] \) to indicate whether the link between the RIS and vehicle \( k \) in the \( l \)-th grid is LoS or non-line-of-sight (NLoS) conditions, as widely adopted in previous studies~\cite{saad, yiqin, you}.  Specifically, $a_{\mathrm{R}, k}[l] = 1$ indicates the LoS state, whereas $a_{\mathrm{R}, k}[l] = 0$ corresponds to the NLoS state.  \( a_{\mathrm{R}, k}[l] \) is influenced by factors such as the propagation environment, vehicle density, and the location of the RIS, which can be commonly approximated as \cite{LAP}
\begin{equation}
    \begin{aligned} \label{eq:h_prob}
        \mathbb{P}(a_{\mathrm{R}, k}[l] = 1) = \frac{1}{1+A_1 \exp(-A_2(\theta_{\mathrm{R}, k}[l]-A_1))},
    \end{aligned}
\end{equation}
where $A_1$ and $A_2$ are constant parameters that depend on the characteristics of environments \cite{you}, \cite{ITU}, and $\theta_{\mathrm{R}, k}[l]$ is the elevation angle between vehicle $k$ and RIS, which can be expressed as
\begin{equation}\label{eq:theta}
    \theta_{\mathrm{R}, k}[l] = \frac{180}{\pi} \times \arctan\left(\frac{h_{\mathrm{R}}-z_k}{d_{\mathrm{R}, k}[l]}\right),
\end{equation}
and $d_{\mathrm{R}, k}[l] \triangleq \sqrt{(x_{\mathrm{R}}-x_k[l])^2+(y_{\mathrm{R}}-y_k[l])^2}$ represents the horizontal distance between them.
The corresponding NLoS probability can be obtained as $\mathbb{P}(a_{\mathrm{R}, k}[l]= 0) = 1-\mathbb{P}(a_{\mathrm{R}, k}[l] = 1)$. Depending on the LoS or NLoS link conditions, the power of the incident signal from vehicle $k$ to RIS can be expressed as
\begin{equation}
    \begin{aligned}
        P_{k}^{in}[l]=a_{\mathrm{R}, k}[l] P_{k}^{\mathrm{L}, in}[l]+(1-a_{\mathrm{R}, k}[l])P_{k}^{\mathrm{N}, in}[l],
    \end{aligned}
\end{equation}
where $P_{k}^{\mathrm{L}, in}[l]$ represents the incident signal power for the LoS state, and $P_{k}^{\mathrm{N}, in}[l] = \xi_{k}P_{k}^{\mathrm{L}, in}[l]$ is the received signal power for the NLoS link, with $\xi_{k}<1$ denoting additional signal attenuation factor accounting for the NLoS propagation~\cite{LAP, yiqin, you}. 

\begin{figure}
    \centering
    \includegraphics[width=1.0\linewidth]{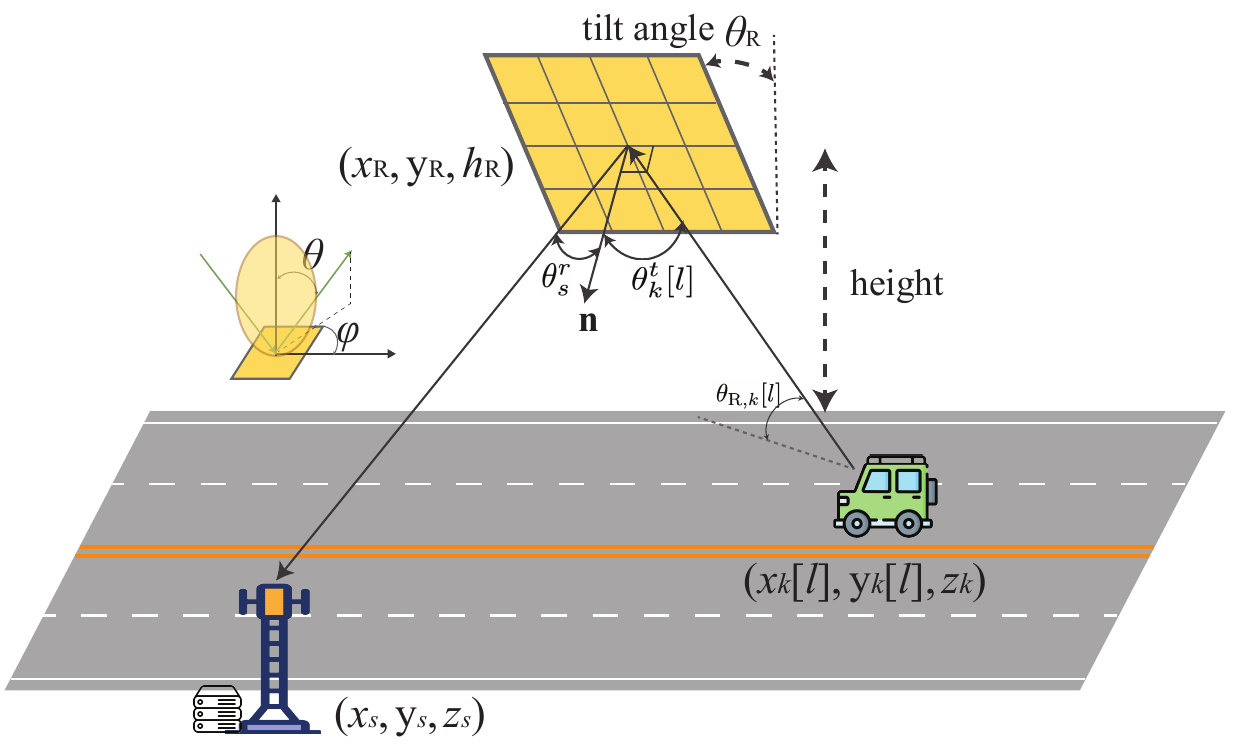}
    \caption{Illustration of a RIS tilted downward to better serve vehicles on the ground.}
    \label{fig: pattern}
\end{figure}

Similarly, we define $c_{\mathrm{R}, s}=1$ and $c_{\mathrm{R}, s}=0$ to represent the binary LoS or NLoS state of the link between RIS and VEC server $s$. The probabilities of the LoS and NLoS channel states, i.e., $\mathbb{P}(c_{\mathrm{R}, s} = 1)$ and $\mathbb{P}(c_{\mathrm{R}, s}=0)$, can be calculated through a similar manner in equations \eqref{eq:h_prob} and \eqref{eq:theta}. Based on this, the received signal power at the VEC server $s$ from vehicle $k$ via the RIS reflection can be given by
\begin{equation}
    \begin{aligned} \label{eq: p_k,s}
        P_{k, s}^{rcv}[l] = c_{\mathrm{R}, s} P_{k, s}^{\mathrm{L}, rcv}[l]+(1-c_{\mathrm{R}, s})P_{k, s}^{\mathrm{N}, rcv}[l],
    \end{aligned}
\end{equation}
where $P_{k, s}^{\mathrm{L}, rcv}[l]$ represents the received power for the LoS connection, and $P_{k, s}^{\mathrm{N}, rcv}[l] = \xi_s P_{k, s}^{\mathrm{L}, rcv}[l]$ is the received power for the NLoS link, with $\xi_s<1$ being the additional signal attenuation factor. 

Considering the cascade vehicle-RIS-server channel, the overall received signal power at server $s$ from vehicle $k$ at the $l$-th grid can be expressed as
\begin{equation}
    \begin{aligned}
        P_{k, s}^{rcv}[l] = \omega_{k,s}[l]\Bar{P}_{k, s}^{rcv}[l],
    \end{aligned}
\end{equation}
where 
\begin{equation}
    \begin{aligned}
        \omega_{k,s}[l]= &a_{\mathrm{R}, k}[l]c_{\mathrm{R}, s}+\xi_k\xi_s(1-a_{\mathrm{R}, k}[l])(1-c_{\mathrm{R}, s})\\
        &+\xi_k(1-a_{\mathrm{R}, k}[l])c_{\mathrm{R}, s} +\xi_s a_{\mathrm{R}, k}[l](1-c_{\mathrm{R}, s}),
    \end{aligned}
\end{equation}
represents the combined LoS and NLoS channel states, and $\Bar{P}_{k, s}^{rcv}[l]$ is the received signal power conditioned on that both the vehicle-RIS link and RIS-server link are LoS, the expression of which will be provided below.

To derive $\Bar{P}_{k, s}^{rcv}[l]$, we first define the normalized power radiation pattern of each RIS element below to measure the RIS gain~\cite{radiation}
\begin{equation}
    F(\theta, \varphi)= \begin{cases}\cos ^3 \theta & \theta \in\left[0, \frac{\pi}{2}\right], \varphi \in[0,2 \pi], \\ 0 & \theta \in\left(\frac{\pi}{2}, \pi\right], \varphi \in[0,2 \pi],\end{cases}
\end{equation}
where $\theta$ and $\varphi$ are the elevation and azimuth angles from the RIS element to a specific transmitting or receiving direction, as shown in Fig. \ref{fig: pattern}. Note that the normalized power radiation pattern is a function of the elevation angle, with the maximum gain in the $\theta = 0$ direction. Let $G_t$ and $G_r$ represent the antenna gains at vehicles and servers, $b$ and $d$ denote the length and width of each RIS element, respectively. 
Without loss of generality, we assume that the peak radiation directions of the transmitting and receiving antenna are aligned with the center of the RIS, and all RIS elements share the same reflection coefficient, i.e., $e^{j\psi} = e^{j\psi_{m}}, \forall{m}=1,\ldots, M_r\times M_l$. The overall received signal power at server $s$ from vehicle $k$ at the $l$-th grid via RIS reflection can be expressed as \cite{Tang}
\begin{equation}
\begin{aligned}
\Bar{P}_{k, s}^{rcv}[l] &=  \frac{P_tG_t G_r G  M_r^2  M_l^2 bd\lambda^2 F(\theta_{k}^t[l], \varphi_{k}^t[l])F\left(\theta_{s}^r, \varphi_{s}^r\right)}{64 \pi^3(d_{k}[l])^\alpha(d_{s})^\alpha} \\
& \times \left\lvert\, \frac{\operatorname{sinc}\left(\frac{M_r \pi}{\lambda}\left(\sin \theta_{k}^t[l] \cos \varphi_{k}^t[l]+\sin \theta_{s}^r\cos \varphi_{s}^r\right) b\right)}{\operatorname{sinc}\left(\frac{\pi}{\lambda}\left(\sin \theta_{k}^t[l] \cos \varphi_{k}^t[l]+\sin \theta_{s}^r \cos \varphi_{s}^r\right) b\right)}\right. \\
& \times\left.\frac{\operatorname{sinc}\left(\frac{M_l \pi}{\lambda}\left(\sin \theta_{k}^t[l] \sin \varphi_{k}^t[l]+\sin \theta_{s}^r\sin \varphi_{s}^r\right) d\right)}{\operatorname{sinc}\left(\frac{\pi}{\lambda}\left(\sin \theta_{k}^t[l] \sin \varphi_{k}^t[l] +\sin \theta_{s}^r \sin \varphi_{s}^r\right) d\right)}\right|^2,
\end{aligned}
\end{equation}
where $P_t$ is the transmit power of the vehicle, $G$ is the gain of the unit RIS element, $\lambda$ is the wavelength, $\alpha$ denotes the average path loss exponent, $d_{k}[l]=\sqrt{(x_\mathrm{R}-x_k[l])^2+(y_\mathrm{R}-y_k[l])^2+(h_\mathrm{R}-z_k)^2}$ is the distance between vehicle $k$ and the RIS, $d_{s}=\sqrt{(x_\mathrm{R}-x_s)^2+(y_\mathrm{R}-y_s)^2+(h_\mathrm{R}-z_s)^2}$ is the distance between RIS and server $s$, and angles $\theta_{k}^t[l], \theta_s^r, \varphi_{k}^t[l]$, and $\varphi_{s}^r$ are given, respectively, by
\begin{equation}
    \begin{aligned}
        \theta_{k}^t[l]  &= \arccos \left( \frac{\cos\theta_{\mathrm{R}}  |y_\mathrm{R}-y_k[l]| + \sin\theta_{\mathrm{R}} |h_\mathrm{R}-z_k|}{d_{k}[l]} \right),\\
        \theta_s^r  &= \arccos \left( \frac{\cos\theta_{\mathrm{R}} |y_\mathrm{R}-y_s| + \sin\theta_{\mathrm{R}} |h_\mathrm{R}-z_s|}{d_{k}[l]} \right),\\
        \varphi_{k}^t[l] &= \arccos \left( \frac{x_k[l] - x_\mathrm{R}}{\sqrt{(x_k[l] - x_\mathrm{R})^2 + \left( \Tilde{d}_1 \right)^2 + \left( \Tilde{d}_2 \right)^2}} \right),\\
        \varphi_{s}^r &= \arccos \left( \frac{x_s - x_\mathrm{R}}{\sqrt{(x_s - x_\mathrm{R})^2 + \left( \Tilde{d}_3 \right)^2 + \left( \Tilde{d}_4 \right)^2}} \right),
    \end{aligned}
\end{equation}
with
\begin{equation}
    \begin{aligned}
        \Tilde{d}_1 =& (y_k[l] - y_\mathrm{R}) - (y_k[l] - y_\mathrm{R}) \cos^2\theta_{\mathrm{R}} \\
        &- (z_k - h_\mathrm{R}) \cos\theta_{\mathrm{R}} \sin\theta_{\mathrm{R}},\\
        \Tilde{d}_2 =& (z_k - h_\mathrm{R}) - (y_k[l] - h_\mathrm{R}) \cos\theta_{\mathrm{R}} \sin\theta_{\mathrm{R}} \\
        &- (z_k - h_\mathrm{R}) \sin^2\theta_{\mathrm{R}},\\
        \Tilde{d}_3 =& (y_s - y_\mathrm{R}) - (y_s - y_\mathrm{R}) \cos^2\theta_{\mathrm{R}} \\
        &- (z_s - h_\mathrm{R}) \cos\theta_{\mathrm{R}} \sin\theta_{\mathrm{R}},\\
        \Tilde{d}_4 =& (z_s - h_\mathrm{R}) - (y_s - h_\mathrm{R}) \cos\theta_{\mathrm{R}} \sin\theta_{\mathrm{R}} \\
        &- (z_s - h_\mathrm{R}) \sin^2\theta_{\mathrm{R}}.
    \end{aligned}
\end{equation}


Based on the channel model above, the data rate from vehicle \( k \) in the \( l \)-th grid to server \( s \), i.e., the ``link capacity" \( \Tilde{R}_{k,s}[l]\), can be calculated through
\begin{equation}
    \Tilde{R}_{k,s}[l] = B\log_2\left(1+\frac{P_{k,s}^{rcv}[l]}{n_0}\right),
\end{equation}
where $B$ is the bandwidth of each channel and $n_0$ is the noise power. Since OFDMA is employed, interference from other vehicles/tasks is not considered herein.

\subsection{Latency Model}
When a vehicle offloads its computing task to a VEC server, the server processes this task and returns the results to the vehicle. Hence, the end-to-end latency consists of the following three components: 1) the communication latency, \( t^{\mathrm{comm}} \), which is the time to transmit the computing task to the server via the RIS-assisted wireless link, 2) the computing latency, \( t^{\mathrm{comp}} \), which is the time to execute the task at the server, and 3) the result transmission latency, \( t^{\mathrm{back}} \), which is the time to transmit the result back to the vehicle via RIS reflection. Since the size of the task execution results is typically small, we follow the common assumption in the literature \cite{MEClatency}, \cite{Cuilatency}, \cite{Zhanglatency} and ignore this part in this paper.

\textbf{1) Computing Latency:}
First, we analyze the required time for task execution. Since servers are often equipped with high-performance GPUs featuring multiple cores or supporting multiple virtual machines (VMs), simultaneous multi-task processing can be conducted~\cite{Huangdelay}. This parallel processing capability significantly enhances the task execution efficiency in VEC systems. Consequently, queuing delays are negligible, and the computing latency for vehicle $k$ can be calculated from
\begin{equation}
    \begin{aligned}
        t^{\mathrm{comp}}_{k} =\frac{D_kF_k}{f},
    \end{aligned}
\end{equation}
where $f$ represents the computing resource allocated to each task, i.e., the number of floating-point operations per second (FLOPS). Without loss of generality, we assume that each task is assigned the same amount of FLOPS\footnote{The optimization of computing resource allocation can be explored in our future work, which is out of the scope of this paper.}.

Due to the limited communication-computing resources, the number of tasks that can be processed simultaneously on a server is subject to
\begin{equation}
    \sum_{k=1}^{K} \lambda_{k, s} \leq C_s, ~\forall s\in \mathcal{S},
\end{equation}
where $C_s=\max\{C_{\mathrm{comp},s}, C_{\mathrm{comm},s}\}$ represents the maximum number of tasks server $s$ can accommodate, with $C_{\mathrm{comp},s}$ denoting the maximum number of computing tasks one server can execute concurrently and $C_{\mathrm{comm},s}$ being the maximum number of communication channels available on server $s$\footnote{\rev{
In this paper, we adopt the OFDMA method, assuming each vehicle is assigned with the same bandwidth (i.e., the same number of subcarriers in OFDMA). Each edge server is allocated with different frequency bands, which therefore does not interfere with each other. This imposes an upper limit on the number of vehicles it can support simultaneously on each edge server: The number of accommodated vehicles cannot exceed the maximum number of available channels on each edge server.
}}.

\textbf{2) Communication Latency:} To ensure the end-to-end latency remains within the tolerable delay, the maximum communication time for uploading a task is
\begin{equation}
    t_{k, \max}^{\mathrm{comm}}= T^{th}_k - t_{k}^\mathrm{comp}, ~\forall k \in \mathcal{K}.
\end{equation}
During this time, vehicle \(k\) can move across up to \(L_k\) complete grids, as shown in Fig. \ref{fig: mobility}, which can be calculated by
\begin{equation}
    L_k = \left\lfloor \frac{t_{k, \max}^{\mathrm{comm}}}{t_k} \right\rfloor,
\end{equation}
where $\lfloor\cdot \rfloor$ denotes the floor operation. Therefore, the maximum amount of data that can be uploaded from vehicle $k$ to server $s$, the summation of the product of link capacity and sojourn time across \(L_k\) grids, is given by
\begin{equation}
    D_{k,s} = \sum_{l=1}^{L_k}t_kB\log_2\left(1+\frac{P_{k,s}^{rcv}[l]}{n_0}\right).
\end{equation}
where $t_k$ is the sojourn time vehicle $k$ in each grid.

To ensure that the task of vehicle $k$ can be successfully offloaded and processed within the required time, the total amount of data that can be supported for uploading, i.e., \( D_{k,s} \), must be no less than the maximum amount \( D_k \), i.e.,
\begin{equation}
    D_{k,s}\geq D_k. \label{eq: data}
\end{equation}
Otherwise, the task cannot be processed within the allowable time. Since \( D_{k,s} \) is a random variable influenced by the randomness of $P_{k,s}^{rcv}[l]$, we count a task to be successful if the probability that the task completion exceeds a certain threshold, i.e.,
\begin{equation}
   \mathbb{P}[D_{k,s}\geq D_k]\geq \eta,  ~\forall k\in \mathcal{K}, ~\forall s\in \mathcal{S}, \label{prob}
\end{equation}
where $\eta$ is the minimum acceptable task completion probability. 
We require that the task assignment decision meet
\begin{equation}
    \begin{aligned}
        \lambda_{k, s}\leq &\mathds{1}{\left(\mathbb{P}\left[D_{k,s}\geq D_k\right] \geq\eta \right)},\\
        &\forall k\in \mathcal{K},~ \forall s\in \mathcal{S}, 
    \end{aligned}
\end{equation}
where $\mathds{1}{(\cdot)}$ is the indicator function. In other words, the task of vehicle $k$ can be assigned to VEC server $s$ only if the task completion probability exceeds the threshold $\eta$ to avoid resource wastage due to task failure. Our design goal, hence, is to maximize the number of completed tasks satisfying the probabilistic constraint (\ref{prob}). Our system model underscores the importance of optimizing RIS placement, as the task completion probability in (\ref{prob}) depends on both task assignment and communication channels, both of which are heavily influenced by the RIS's altitude and tilt angle.




\section{Task Throughput Maximization and Proposed Algorithm}
\label{sec: problem and alg}
In this section, we first present an optimization problem aiming at maximizing the average task throughput by jointly optimizing RIS placement and task offloading strategy. 
\newrev{Next, we provide our two-layer optimization framework, i.e., optimizing task offloading in the inner layer and RIS positioning in the outer layer.}
Finally, we analyze the complexity of the proposed algorithm.

\subsection{Problem Formulation}
Considering the RIS-assisted VEC system in Section \ref{sec:system}, we aim to maximize the task throughput, \rev{i.e., the number of computing tasks that can be offloaded and successfully computed.} 
To achieve this, we propose to jointly optimize the task offloading strategy \(\lambda_{k,s}\), the RIS altitude \(h_{\mathrm{R}}\), and tilt angle \(\theta_{\mathrm{R}}\). 
To optimize the long-term performance, we need to design RIS placement for improving the average task throughput across all instances. Let $\mathcal{N}=\{1, 2, \dots, N\}$ be the sequence of instances, $\mathcal{K}{[n]}=\{1, 2, \dots, K{[n]}\}$ denote the set of vehicles in instance $n$, and the vector \( \boldsymbol{\lambda}{[n]}\triangleq [\lambda_{1,1}[n], \ldots, \lambda_{K[n],S}[n]]^T\) represent the task assignment decisions during instance $n$. The corresponding optimization problem can be formulated as follows
\begin{subequations} \label{p:original}
    \begin{align}
        \max_{\boldsymbol{\lambda}{[n]}, h_{\mathrm{R}}, \theta_{\mathrm{R}}} &\frac{1}{N}\sum_{n=1}^N\sum_{k=1}^{K{[n]}}\sum_{s=1}^S \lambda_{k,s}[n] \label{eq: p1_obj}\\
        \mathrm{s. t.}~~~ &\sum_{s=1}^S \lambda_{k, s}[n] \leq 1, ~\forall k\in \mathcal{K}{[n]},~\forall n\in\mathcal{N}, \label{eq: p1_task}\\
        &\sum_{k=1}^{K{[n]}} \lambda_{k, s}[n] \leq C_s, ~\forall s\in \mathcal{S}, ~\forall n\in\mathcal{N},\label{eq: p1_server}\\
        & \lambda_{k, s}[n] \leq \mathds{1}{\left(\mathbb{P}\left[D_{k,s}[n]\geq D_k\right] \geq\eta \right)}, \nonumber\\
        &~~~~~~~~~~~~\forall k\in \mathcal{K}{[n]}, ~\forall s\in \mathcal{S},~\forall n\in\mathcal{N}, \label{eq: p1_delay}\\
        &\lambda_{k, s}[n]\in \{0,1\}, ~\forall k\in \mathcal{K}{[n]}, ~\forall s\in \mathcal{S}, ~\forall n\in\mathcal{N},\label{eq: p1_lambda}\\
        &H_{\min} \leq h_{\mathrm{R}} \leq H_{\max}, \label{eq: p1_h}\\
        &0 \leq \theta_{\mathrm{R}} \leq \frac{\pi}{2}. \label{eq: p1_theta}
    \end{align}
\end{subequations}
Constraint \eqref{eq: p1_task} implies that each computing task can be assigned to at most one VEC server, and Constraint \eqref{eq: p1_server} ensures that the number of tasks processed concurrently does not exceed the maximum capacity of edge servers due to communication-computing resource limitations. Constraint \eqref{eq: p1_delay} enforces that the task completion probability from a vehicle to a server exceeds a predefined threshold once assigned. \eqref{eq: p1_lambda} is the constraint for the binary task assignment decisions. Constraint \eqref{eq: p1_h} sets the maximum and minimum values for the RIS altitude, while Constraint \eqref{eq: p1_theta} defines the range of the tilt angle.

\newrev{Problem \eqref{p:original} is highly challenging since it is mixed-integer nonlinear programming, where the non-linearity stems from Constraint \eqref{eq: p1_delay}, with $D_{k,s}[n]$ being a highly non-convex function of RIS altitude $h_R$ and tilt angle $\theta_R$.}
To address the joint task offloading and RIS placement problem, we propose a two-layer algorithm. For ease of presentation, we first elaborate on how to optimize the task assignment strategy for each instance in the inner layer. Subsequently, we show how to optimize the placement decision in the outer layer, executing the task assignment algorithm at each step. The details are given below.



\subsection{Task Assignment Optimiztion}
\label{sec:task_assign}
Assuming a fixed RIS position, we determine the task assignment for each instance. 
Specifically, for a particular instance $[n]$, the optimization problem can be reformulated into
\begin{subequations} \label{assignment}
    \begin{align}   
        \max_{\boldsymbol{\lambda}[n]}~ &\sum_{k=1}^{K[n]}\sum_{s=1}^S \lambda_{k,s}[n] \label{eq: p2_obj}\\
        \mathrm{s. t.}~ &\sum_{s=1}^S \lambda_{k, s}[n] \leq 1, ~\forall k\in \mathcal{K}[n], \\
        &\sum_{k=1}^{K[n]} \lambda_{k, s}[n] \leq C_s, ~\forall s\in \mathcal{S}, \label{eq:p2_sever} \\
        & \lambda_{k, s}[n] \leq \mathds{1}{\left(\mathbb{P}\left[D_{k,s}[n]\geq D_k\right] \geq\eta \right)}, \nonumber \\
        &~~~~~~~~~~~~~\forall k\in \mathcal{K}[n], ~\forall s\in \mathcal{S}, \\
        &\lambda_{k, s}[n]\in \{0,1\}, ~\forall k\in \mathcal{K}[n], ~\forall s\in \mathcal{S},
    \end{align}
\end{subequations}
which is an integer linear programming, known to be computationally challenging in general with an exponentially increasing solution space. Fortunately, the optimization problem can be transformed into a \textbf{linear program} (LP) where the relaxed constraints naturally enforce integer solutions without additional constraints~\cite{Taylor}.


\textbf{LP Transformation.} We observe that the task assignment problem is equivalent to a bipartite matching, with one vehicle set and one server set. The edges between nodes in these two sets denote task assignments. We define two node edge adjacency matrices, $\mathbf{E}_{\mathrm{V}}[n]\in \mathbb{R}^{K[n]\times (K[n]\times S)}$ and $\mathbf{E}_{\mathrm{S}}[n]\in \mathbb{R}^{S\times (K[n] \times S)}$. In matrix $\mathbf{E}_{\mathrm{V}}[n]$, each row corresponds to a vehicle, and every column is an edge indicating the vehicle-server matching. Specifically, $\mathbf{E}_{\mathrm{V}}[n](i,j)=1$ means vehicle $i$ offloads computing task by edge $j$, while $\mathbf{E}_{\mathrm{V}}[n](i,j)=0$ otherwise. Similarly, every row in matrix $\mathbf{E}_{\mathrm{S}}[n]$ corresponds to a server, and each column corresponds to an edge signifying the vehicle-server matching. Specifically, $\mathbf{E}_{\mathrm{S}}[n](i,j)=1$ indicates that server $i$ is assigned task via edge $j$, while $\mathbf{E}_{\mathrm{S}}[n](i,j)=0$ otherwise. 
According to this definition, the node edge matrices have the following properties: The matrix $\mathbf{E}_{\mathrm{V}}[n]\boldsymbol{\lambda}[n]$ produces a vector in $\mathbb{R}^{K[n]}$, representing the number of tasks offloaded by each vehicle, which is at most one in our case. Similarly, $\mathbf{E}_{\mathrm{S}}[n]\boldsymbol{\lambda}[n]$ returns a vector in $\mathbb{R}^{S}$, representing the number of tasks assigned to each server. With these notations,  the optimization problem with decision vector $\boldsymbol{\lambda}[n]$ can be rewritten in matrix representation as follows
\begin{subequations}\label{p:linear_lambda}
    \begin{align}
        \max_{\boldsymbol{\lambda}[n]}~ ~&\mathbf{1}^T\boldsymbol{\lambda}[n]\\
        \mathrm{s. t.}~~ & \mathbf{W}[n]\boldsymbol{\lambda}[n]\preceq \mathbf{v}[n],
    \end{align}
\end{subequations}        
where
\begin{equation}
\mathbf{W}[n]=\left(\begin{array}{c}
\mathbf{E}_{\mathrm{V}}[n] \\
\mathbf{E}_{\mathrm{S}}[n] \\
\mathbf{I} \\
\mathbf{I} \\
-\mathbf{I}
\end{array}\right), \mathbf{v}[n]=\left(\begin{array}{c}
\mathbf{1} \\
\mathbf{c}_s \\
\mathbf{q}[n] \\
\mathbf{1} \\
\mathbf{0}
\end{array}\right),
\end{equation}
and $\mathbf{I}$ is the identity matrix, $\mathbf{1}$ is the vector with all ones, $\mathbf{c}_s\triangleq [C_1,\ldots,C_S]^T$ is the vector denoting the maximum number of tasks that servers can handle, $\mathbf{q}[n] \triangleq [\mathds{1}{\left(\mathbb{P}\left[D_{1,1}\geq D_1\right] \geq\eta \right)}, \ldots,\mathds{1}{\left(\mathbb{P}\left[D_{K[n],S}\geq D_{K[n]}\right] \geq\eta\right)}]^T$ is the indicator vector representing whether the task completion probability can be satisfied or not, and $\mathbf{0}$ is the vector with all zeros.

We have the following definition and proposition.

\begin{algorithm}[t] 
	\caption{Optimal Task Offloading}
 \label{alg: Alg1}
	\LinesNumbered 
	\KwIn{The computing task parameter set $\mathcal{T}_k(D_k, F_k, T^{th}_k)$, the maximum number of supporting tasks $C_s$, the task completion probability threshold $\eta$, RIS placement $(h_{\mathrm{R}}, \theta_{\mathrm{R}})$}
	\KwOut{The optimal task offloading $\boldsymbol{\lambda}[n]^\star$, $\forall n \in \mathcal{N}$}
 \For{$n \in \mathcal{N}$}{
 Calculate the task completion probability indicator $\mathds{1}{\left(\mathbb{P}\left[D_{k,s}[n]\geq D_k\right] \geq\eta \right)}$, $\forall k \in \mathcal{K}[n], \forall s \in \mathcal{S}$\;
 Obtain the task offloading strategy $\boldsymbol{\lambda}[n]^\star$ by solving the LP problem \eqref{p:linear_lambda}\;
 }
Return the optimal task offloading $\boldsymbol{\lambda}[n]^\star$, $\forall n \in \mathcal{N}$
\end{algorithm}

\begin{definition}
    A totally unimodular matrix is defined as a matrix of which every square non-singular submatrix is unimodular. Equivalently, every square submatrix has determinant $0, +1$, or $-1$.
\end{definition}

\begin{proposition}
    The integer matrix $\mathbf{W}[n]$ is totally unimodular.
\end{proposition}

\begin{proof}
     The node edge adjacency matrices, $\mathbf{E}_{\mathrm{V}}[n]$ and $\mathbf{E}_{\mathrm{S}}[n]$, are totally unimodular by construction \cite{Princeton}. If a matrix is totally unimodular, any matrices derived from it by appending $\mathbf{I}$ or $-\mathbf{I}$ remains also unimodular according to \cite{unimodular}. Hence, the matrix $\mathbf{W}[n]$ retains this property as well.
\end{proof}

Given that the constraint matrix $\mathbf{W}[n]$ is totally unimodular and the entries in the vector $\mathbf{v}[n]$ are all integers, \rev{according to Theorem 13.2 in \cite{unimodular}}, the vertices of the convex polytope defined by $\mathbf{W}[n]\boldsymbol{\lambda}[n]\preceq \mathbf{v}[n]$ have integer coordinates. 
\rev{Since the optimal solution of a linear program lies at a vertex of its feasible region, at least one optimal solution to our problem will be binary.}
As a result, we can reformulate and solve the linear program efficiently, and the relaxation, i.e., $\lambda_{k, s}[n]\in [0,1]$, preserves the integrality of the optimal solution while satisfying all relevant constraints.
The algorithm is outlined in Algorithm \ref{alg: Alg1}. Additionally, utilizing an LP solver for solving such a linear program proves to be very efficient \cite{Boyd}. 
This method yields optimal task assignments with low computational complexity, \newrev{making it possible to obtain the optimal task offloading in each iteration of searching RIS positions, as detailed below.}


\subsection{RIS Placement Optimization}
As alluded to earlier, raising the RIS altitude increases the probability of establishing an LoS connection while also increasing the path attenuation of the vehicle-RIS-server channel. Since both the vertical position and the tilt angle of the RIS affect the performance, the resulting placement problem becomes highly intractable. In what follows, we provide a detailed explanation of the \textbf{Intractability of the Problem}.

\begin{itemize}
    \item \textbf{Implicit Relationship:} RIS placement has an indirect impact on task throughput, making it difficult to deduce the exact influence of RIS placement on our objective function. Additionally, the RIS position also impacts the task completion probability constraints in an implicit manner, adding another layer of complexity.

    \item \textbf{Non-continuity:} RIS positioning influences the probability of LoS or NLoS channels between RIS and vehicles or servers, resulting in a probabilistic constraint with an indicator function, which is discontinuous. Therefore, there is no informative gradient available for gradient-based methods.

    \item \textbf{Trigonometric Dependence:} The RIS altitude and tilt angle affect its gain, i.e., the power radiation pattern, by changing the elevation angles from the RIS to vehicles and servers, i.e., $\theta_{k}^t[l]$ and $\theta_{s}^r$. These angles involve inverse trigonometric functions. In addition, the probability that links are LoS or NLoS also contains trigonometric expressions as in equation \eqref{eq:theta}, adding more complexity to the optimization.
    
\end{itemize}
The aforementioned issues prevent us from acquiring a closed-form solution to RIS placement. Fortunately, the altitude $h_{\mathrm{R}}$ and tilt angle $\theta_{\mathrm{R}}$ of the RIS are both bounded. Therefore, we first provide an optimal solution through grid search. Then, we propose a low-complexity solution using a hill climbing algorithm. 


\begin{algorithm}[t] 
	\caption{Optimal Design for RIS Placement}
 \label{alg: AlgOP}
	\LinesNumbered 
	\KwIn{The computing task parameter set $\mathcal{T}_k(D_k, F_k, T^{th}_k)$, initial position of vehicle $k$ $(x_k[1], y_k[1], z_k)$, maximum number of supporting tasks $C_s$, task completion probability threshold $\eta$, channel environment parameters $A_1$ and $A_2$, length of each grid $\Delta_d$, velocity $v_k$, antenna gains $G_t$ and $G_r$, length and width of RIS element $b$ and $d$, additional signal attenuation factors $\xi_k$ and $\xi_s$, bandwidth $B$, and the feasible set for RIS placement $\mathcal{U}$}
	\KwOut{The optimal RIS placement $(h_{\mathrm{R}}, \theta_{\mathrm{R}})^\star$ and optimal task offloading $\boldsymbol{\lambda}[n]^\star$, $\forall n \in \mathcal{N}$}
 \For{$(h_{\mathrm{R}}, \theta_{\mathrm{R}})^{(i)} \in 
 \mathcal{U}$}{ 
 Optimize task offloading $\boldsymbol{\lambda}[n]^\star$, $\forall n \in \mathcal{N}$, based on Algorithm \ref{alg: Alg1}\;
 Calculate the task throughput by $R({(h_{\mathrm{R}}, \theta_{\mathrm{R}})^{(i)}}) = \frac{1}{N}\sum_{n=1}^N\sum_{k=1}^{K[n]}\sum_{s=1}^S \lambda_{k,s}[n]$\;
   \If{$R({(h_{\mathrm{R}}, \theta_{\mathrm{R}})^{(i)}}) > R^\star$}
   {
   Update task throughput $R^\star \gets R({(h_{\mathrm{R}}, \theta_{\mathrm{R}})^{(i)}}) $\;
   Update placement $(h_{\mathrm{R}}, \theta_{\mathrm{R}})^{\star} \gets (h_{\mathrm{R}},\theta_{\mathrm{R}})^{(i)}$\;}
   }
Return the optimal RIS placement $(h_{\mathrm{R}}, \theta_{\mathrm{R}})^{\star}$ and optimal task offloading $\boldsymbol{\lambda}[n]^\star$, $\forall n \in \mathcal{N}$ 
\end{algorithm}

\rev{
\textbf{I. Grid Search for RIS Placement:}} Since the solution space is only two-dimensional, we first determine the optimal placement pair $(h_{\mathrm{R}}^{\star}, \theta_{\mathrm{R}}^\star)$ through grid search, which is computationally feasible in many situations. According to constraints \eqref{eq: p1_h} and \eqref{eq: p1_theta}, the feasible set can be given by
\begin{equation}
     \mathcal{U} = \mathcal{H} \times \Theta=\left\{ (h_{\mathrm{R}}, \theta_{\mathrm{R}}) \mid h_{\mathrm{R}} \in  \mathcal{H}, \theta_{\mathrm{R}} \in \Theta \right\},
\end{equation}
where $\mathcal{H}$ and $\Theta$ are the sets containing all possible values for $h_{\mathrm{R}}$ and $\theta_{\mathrm{R}}$, respectively, and can be expressed as
\begin{equation}
    \begin{aligned}
        &\mathcal{H} = \{H_{\min}, H_{\min}+\Delta h, H_{\min}+2\Delta h, ..., H_{\max}\},\\
        &\Theta = \{0, \Delta \theta, 2 \Delta \theta, ..., \frac{\pi}{2}\},
    \end{aligned}
\end{equation}
where  $\Delta h$ and $\Delta \theta$ are the step sizes of grid search. The RIS altitude and tilt angle are then found by searching within this feasible set, with each step of search calculating the task throughput by solving \eqref{p:linear_lambda} via Algorithm \ref{alg: Alg1}. \rev{When the search step size is sufficiently small, the obtained RIS placement can be considered optimal, serving as the optimal baseline.} The procedure is summarized in Algorithm \ref{alg: AlgOP}.

\begin{algorithm}[t] 
	\caption{HC Algorithm for RIS Placement}
 \label{alg: AlgHC}
	\LinesNumbered 
	\KwIn{The computing task parameter set $\mathcal{T}_k(D_k, F_k, T^{th}_k)$, initial position of vehicle $k$ $(x_k[1], y_k[1], 0)$, maximum number of supporting tasks $C_s$, task completion probability threshold $\eta$, channel environment parameters $A_1$ and $A_2$, length of each grid $\Delta_d$, velocity $v_k$, antenna gains $G_t$ and $G_r$, length and width of RIS element $b$ and $d$, additional signal attenuation factors $\xi_k$ and $\xi_s$, bandwidth $B$, feasible set for RIS placement $\mathcal{U}$, number of particles $J$, threshold for determining stopping $\delta$, maximum number of iterations $N_{\mathrm{ite}}$}
	\KwOut{The RIS placement $(h_{\mathrm{R}}, \theta_{\mathrm{R}})^\star$ and optimal task offloading $\boldsymbol{\lambda}[n]^\star$, $\forall n \in \mathcal{N}$} 
 Initialize population $\mathcal{P}$ of $J$ particles\; 
 $\epsilon_{\max} = \infty$, $i=1$\;
 \While{$\epsilon_{\max} > \delta$ \text{and} $i \leq N_{\mathrm{ite}}$}{
 $i \gets i +1 $\;
 \For{$\rho_j \in \mathcal{P}$}{
 Optimize task offloading $\boldsymbol{\lambda}[n]^\star$, $\forall n \in \mathcal{N}$, based on Algorithm \ref{alg: Alg1}\;
 Calculate the task throughput by $R(\rho_j) = \frac{1}{N}\sum_{n=1}^N\sum_{k=1}^{K[n]}\sum_{s=1}^S \lambda_{k,s}[n]$\;
  Select random particle $\mu \neq \rho_j$\;
 \For{$\kappa \in \{1, 2\}$}{
    $\epsilon_{\max} \gets |\rho_j(\kappa) - \mu(\kappa)|$\;
 Generate $r \in [-\epsilon_{\max}, +\epsilon_{\max}]$\;
 ${\rho_j}'(\kappa) \gets \rho_j(\kappa) +r$\; 
 }
 Optimize task offloading $\boldsymbol{\lambda}[n]^\star$, $\forall n \in \mathcal{N}$, based on Algorithm \ref{alg: Alg1}\;
 Calculate the task throughput by $R(\rho_j') = \frac{1}{N}\sum_{n=1}^N\sum_{k=1}^{K[n]}\sum_{s=1}^S \lambda_{k,s}[n]$\;
 \If {$R(\rho_j') > R(\rho_j)$}{
$\rho_j \gets \rho_j'$\;
}
}
}
Return RIS placement $(h_{\mathrm{R}}, \theta_{\mathrm{R}})^{\star}$ and the optimal task offloading $\boldsymbol{\lambda}[n]^\star$, $\forall n \in \mathcal{N}$
\end{algorithm}

\begin{remark}
     Since the solution space is two-dimensional and RIS placement is conducted offline, the computational complexity remains practical for problems with reasonable step sizes, as demonstrated in our simulations.
\end{remark}


\textbf{II. HC-based RIS Placement:} 
In addition to grid search, other methods, such as genetic algorithms (GA) \cite{GA}, can be employed as heuristic approaches to address the placement problem \rev{with continuous variables}. However, these heuristics are often time-consuming and tend to lead to sub-optimal solutions. To swiftly locate a near-optimal solution, we herein propose a hill climbing algorithm as a heuristic. The algorithm is motivated by the relationship between the optimal task throughput and RIS placement, where the task throughput forms only one “peak” in terms of RIS placement, i.e., tilt angle and altitude. The visualization can be found in Fig. \ref{fig:throughput_system}, which will be presented in our simulations. Intuitively, raising the RIS can establish connections to more servers but at the price of increasing path attenuation, whereas moving the tilt angle away from the optimal angle decreases the reflective gain. Both the experiments and our intuitions suggest that there is probably only one peak or a local optimum in the performance curves, which is the case where the hill climbing algorithm can find the globally optimal solution.

As outlined in Algorithm \ref{alg: AlgHC}, the hill climbing algorithm begins with initializing multiple particles (Line 1). Then, it makes the random movements of these particles in each round (Lines 9-13) and records new positions if there are improvements (Lines 16-18). The algorithm terminates if no further changes can be made, implying converging to at least a local optimum. Since the efficiency of hill climbing algorithms is influenced greatly by the chosen step size, we devise the algorithm with a self-adaptive step size. In this scheme, the temporary neighborhood of the point $\rho_j$ is determined by the distance between itself and a randomly selected sample $\mu$ (Lines 8-13). In the beginning, this distance is likely to be large because the initial population is uniformly distributed over the search space. As the search progresses, each point gravitates toward a local optimum, naturally reducing the step size. In addition, if $\rho_j$ and $\mu$ are located in different clusters, $\rho_j$ has a chance to escape its local optimum. Thus, our algorithm also maintains a certain level of performance robustness even with multiple local optimums \cite{ASHC}.

We will evaluate our HC-based RIS placement in Section \ref{sec: simulation}. Although it is challenging to prove the global optimality theoretically, our experiments demonstrate that it achieves the global optimum in all our experimental settings. Moreover, it uses much lower running time than other heuristic baselines, such as GA-based algorithms.

\subsection{RIS Phase-shift Optimization}
\label{sec: phase-shift}
Once the RIS is optimally placed, its phase shift is adjusted in real-time. Specifically, an efficient algorithm is employed to optimize the phase shift based on the current number of vehicles and their locations, facilitating better communication and computation load balancing.
Numerous studies have explored RIS phase shift optimization, including the classical semidefinite relaxation (SDR) technique \cite{Qingqing}, majorization-minimization (MM) method \cite{MM}, and manifold optimization \cite{Manifold}. The phase shift design can further improve the performance of the RIS-assisted VEC system.

\subsection{Computational Complexity}
We propose a two-layer algorithm to optimize both RIS placement and task offloading strategy. In the outer layer, a grid search method is employed to find the optimal RIS placement. This process has a computational complexity of $\mathcal{O}(U)$, where $U$ represents the total number of potential positions for the RIS. To enhance efficiency, we introduce a hill climbing algorithm inspired by our insights, with a complexity of $\mathcal{O}(N_{\mathrm{g}}J)$, where $N_{\mathrm{g}}$ is the number of iterations until convergence and $J$ represents the population size.

Once the RIS position is determined in the outer layer, the problem is reduced to a deterministic task offloading optimization in the inner layer. Here, we apply a linear optimization technique to determine the optimal task offloading strategy. Specifically, we can utilize existing solvers that implement the interior-point method to solve the optimization problem, with a typical complexity of $\mathcal{O}((K \times S)^3)$. Consequently, the overall computational complexity for the optimal RIS placement and task offloading process is $\mathcal{O}(U N (K \times S)^3)$. When adopting the hill climbing algorithm for RIS placement, the overall complexity becomes $\mathcal{O}(N_{\mathrm{g}} J N (K \times S)^3)$.
\rev{The low complexity of the algorithm, combined with the parallel processing capability of VEC servers, ensures the scalability of our proposed framework even under heavy computing workloads.}


\section{Performance Evaluation}
\label{sec: simulation}

In this section, we provide extensive simulation results to demonstrate the effectiveness of our proposed RIS placement and task offloading design.  

\begin{table*}[!t]
\centering
\caption{\rev{Parameter settings for simulations}}
\label{table_para}
\renewcommand{\arraystretch}{1.4}
\setlength{\tabcolsep}{2mm}
\rev{
\begin{tabular}{|c|c|c|c|}
\hline
{Center frequency} &  {$f$ = 5.9 GHz}&  {Bandwidth}& {$B$ = 20 MHz}   \\  \hline
Antenna gain & {$G_tG_r = 100$}  & {RIS element gain}  & {$G=8$}   \\ \hline
Path loss exponent   & {$\alpha = 2.7$}  & {Attenuation factor for NLoS channel}  & {\( \xi_k = \xi_s = -20 \) dB } \\ \hline    
{Urban channel environment factors}  & {\( A_1 = 11.95 \), \( A_2 = 0.136 \)}  & {Noise power} &\( n_0 = -100 \) dBm  \\ \hline
Task deadline & {$T^{th}_k$ = 0.1 s}  & {Computing capability for each task}  & {20 TFLOPS} \\ \hline
Number of supporting tasks & {$C_s=4$}  & {Task completion probability}  & {$\eta$ = 0.75 } \\ \hline
\end{tabular}
}
\end{table*}

\subsection{Simulation Settings}
For ease of illustration, we consider a straight four-lane road from -100 m to 100 m with two lanes in each direction, with each lane 4 m wide. We use SUMO to generate vehicle mobility and arrivals, where vehicles travel at speeds between 40 km/h and 72 km/h (for urban roads), \rev{and 80 km/h - 130 km/h (for highway)}. We set the maximum acceleration to 2.0 m/s$^2$, the maximum deceleration (braking) to 3.0 m/s$^2$ based on the Krauss model \cite{Krauss}, and the vehicle arrival rate to 0.7/s.  Assuming the center of the considered road segment is $(0, 0, 0)$, the RIS is located at $(0, -12, h_{\mathrm{R}})$. Unless specified otherwise, four VEC servers are equidistant along the opposite side of the road from the RIS, with the transceivers at an altitude of 6 meters. The maximum and minimum values for the RIS altitude are $H_{\min} =0$ m and $H_{\max}=90$ m.
\rev{Once the RIS is optimally placed, its phase shift is designed using manifold optimization \cite{Manifold}.}

\begin{figure}[t]
    \centering
    \begin{subfigure}[b]{0.42\textwidth}
    \includegraphics[width=\textwidth]{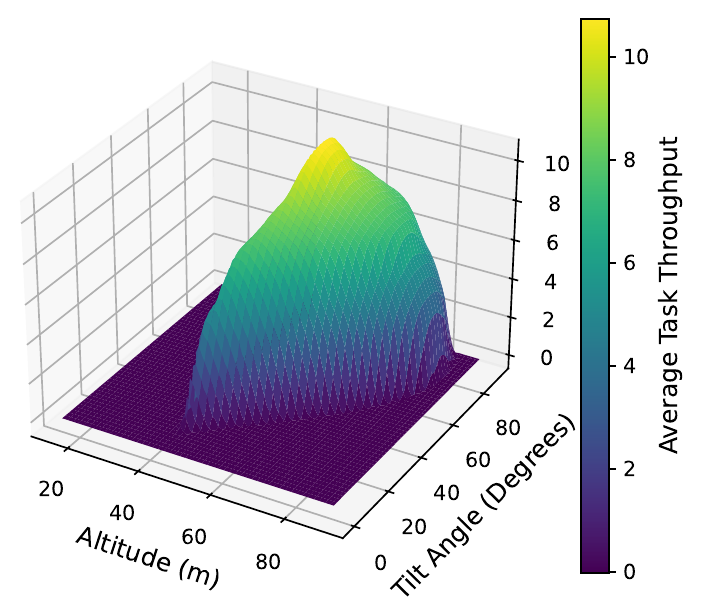}
    \caption{Average task throughput versus RIS placement (4 servers, each supporting up to 3 computing tasks).}
    \label{fig:throughput_deployment_new}
\end{subfigure}
    \hfill
\begin{subfigure}[b]{0.42\textwidth}
    \includegraphics[width=\textwidth]{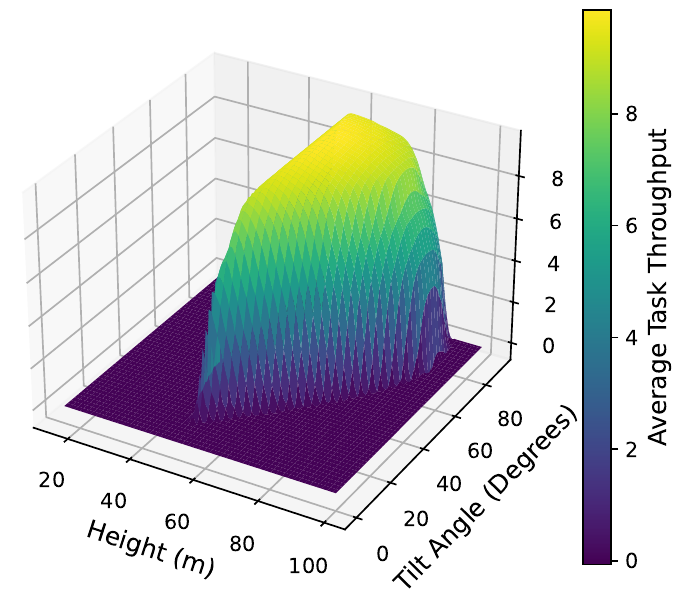}
    \caption{Average task throughput versus RIS placement (6 servers, each supporting up to 2 computing tasks).}
    \label{fig:throughput_deployment_c3_new}
\end{subfigure}
    \caption{Average task throughput versus RIS placement under different VEC server conditions ($\alpha$ = 2.8, $\eta$ = 0.75).}
    \label{fig:throughput_system}
    \vspace{-0.1 cm}
\end{figure}

Simulations are conducted at a center frequency of 5.9 GHz and the bandwidth is $B = 20$ MHz \cite{CV2X}. The antenna gains are set to \( G_tG_r = 100 \), and the gain for each RIS element is \( G = 8 \) \cite{Tang}. The length and width of each RIS element are \( b = d = \frac{\lambda}{5} \), with $200\times 200$ elements \cite{Xiaowen}. 
The path loss exponent $\alpha$ is set to 2.7, and the additional attenuation factor for the NLoS channel is \( \xi_k = \xi_s = -20 \) dB \cite{saad}. Parameters for the urban channel environment are set to \( A_1 = 11.95 \) and \( A_2 = 0.136 \), with a noise power of \( n_0 = -100 \) dBm \cite{saad}. $N=500$ instances are generated. 
For the computing task, we consider object detection, one of the most widely used applications in autonomous driving, and adopt YOLOv7 to perform the detection.
Each vehicle generates a computing task comprising 2 samples, with each sample being a $640 \times 640$ color image \rev{\cite{YOLO}}. The task deadline is $T^{th}_k$ = 0.1 second \rev{\cite{0.1s}}, and the computing complexity per data sample is \( F_k\) = 89.7 GFLOPs \cite{IDOD}. The computing capability allocated for each task is 20 TFLOPS \rev{\cite{Zheng}}. Each server can support up to $C_s=4$ tasks. The task completion probability $\eta$ is set to 0.75.
\rev{The key parameters are summarized in Table \ref{table_para}.}

\begin{figure*}[t]
    \centering
\begin{subfigure}[b]{0.24\textwidth}
    \includegraphics[width=\textwidth]{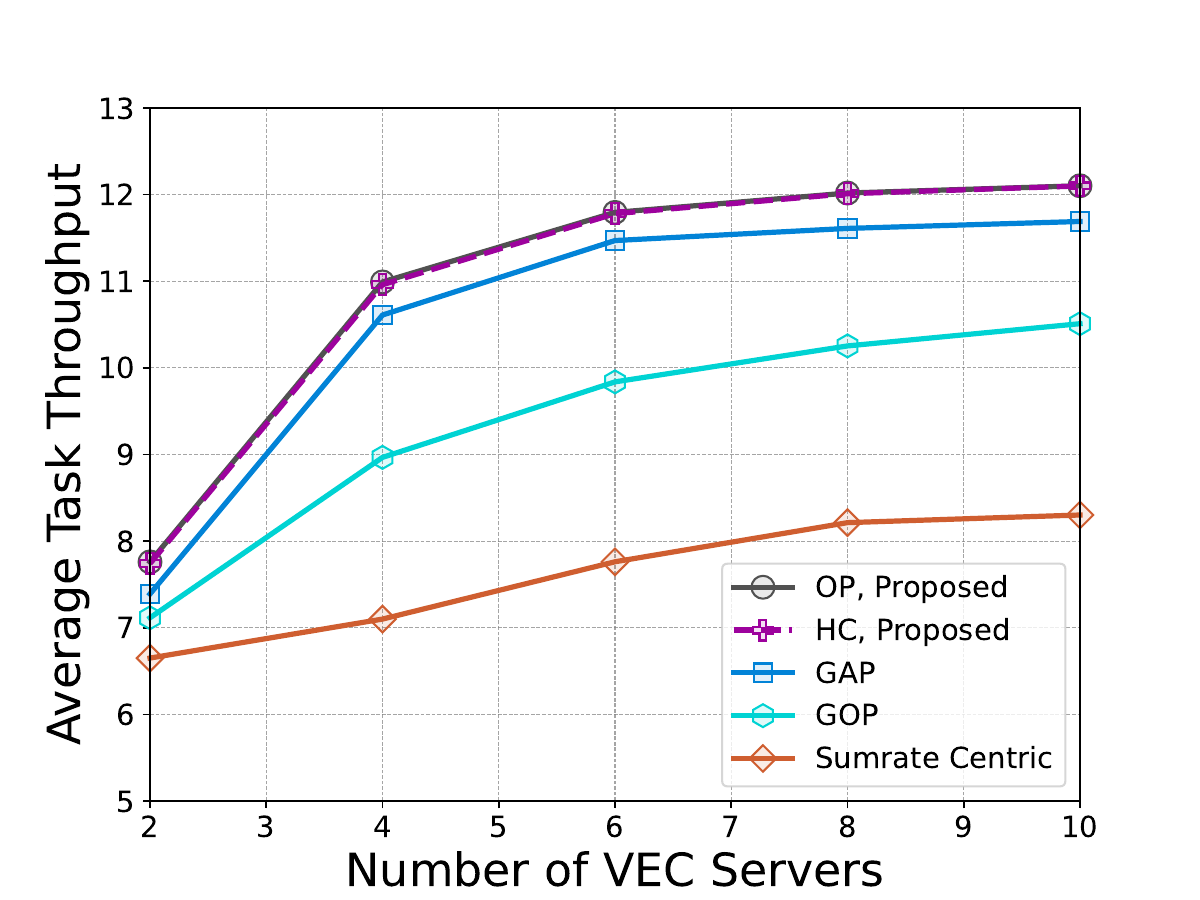}
    \caption{Vehicle arrival rate set to 0.7/s \rev{in urban roads}.}
    \label{fig:num_es}
\end{subfigure}
    \hfill
\begin{subfigure}[b]{0.24\textwidth}
    \includegraphics[width=\textwidth]
    {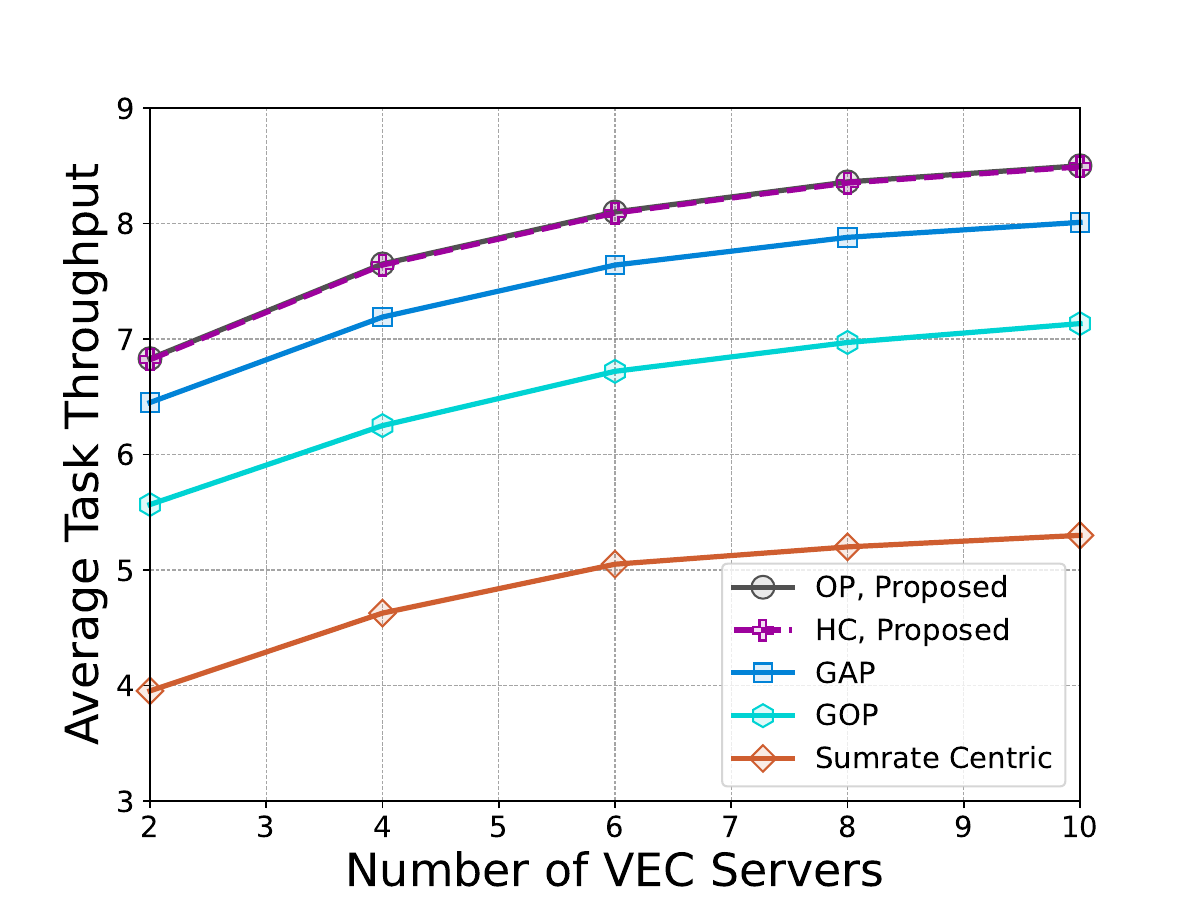}
    \caption{Vehicle arrival rate set to 0.5/s \rev{in urban roads}.}
    \label{fig:num_es_prob5}
\end{subfigure}
    \hfill
\begin{subfigure}[b]{0.24\textwidth}
    \includegraphics[width=\textwidth]{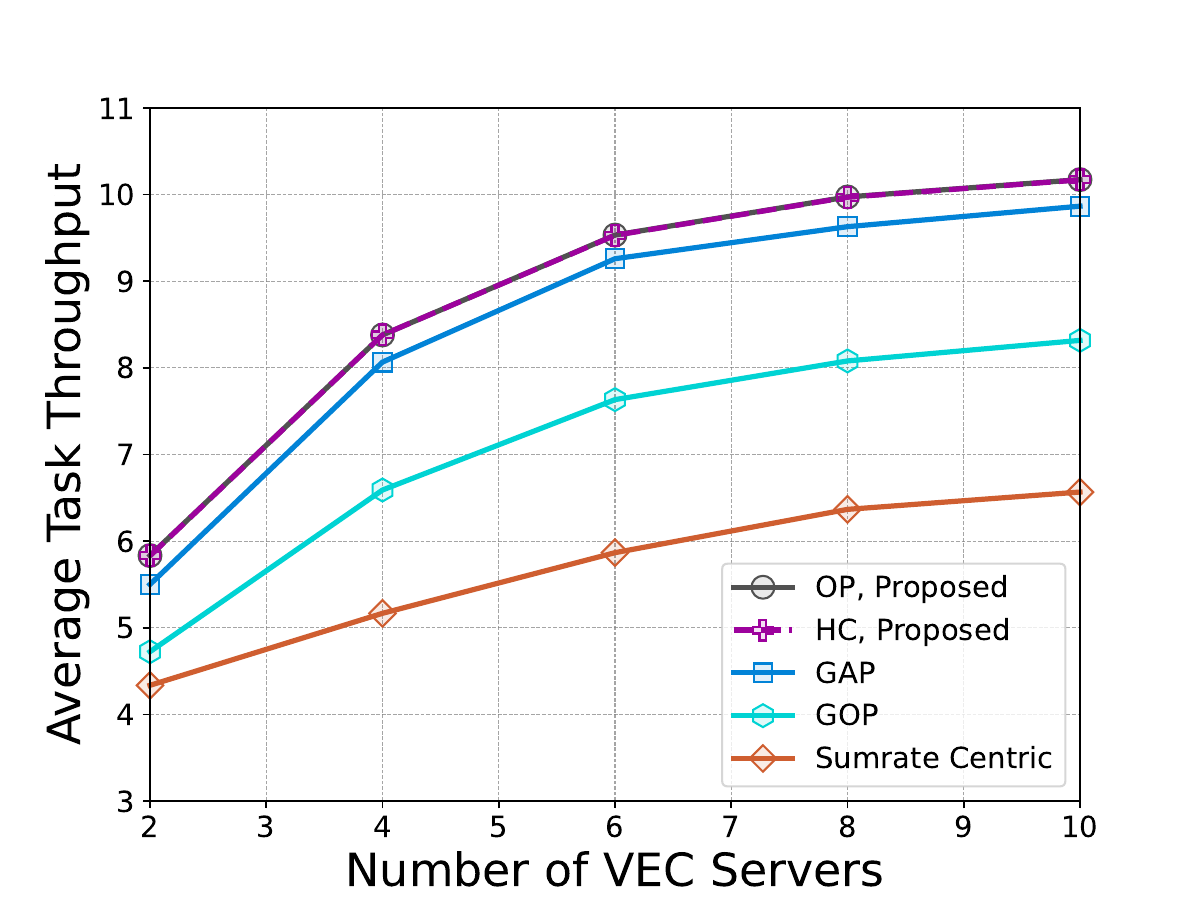}
    \caption{\rev{Vehicle arrival rate set to 0.7/s in highway environments}.}
    \label{fig:num_es}
\end{subfigure}
    \hfill
\begin{subfigure}[b]{0.24\textwidth}
    \includegraphics[width=\textwidth]
    {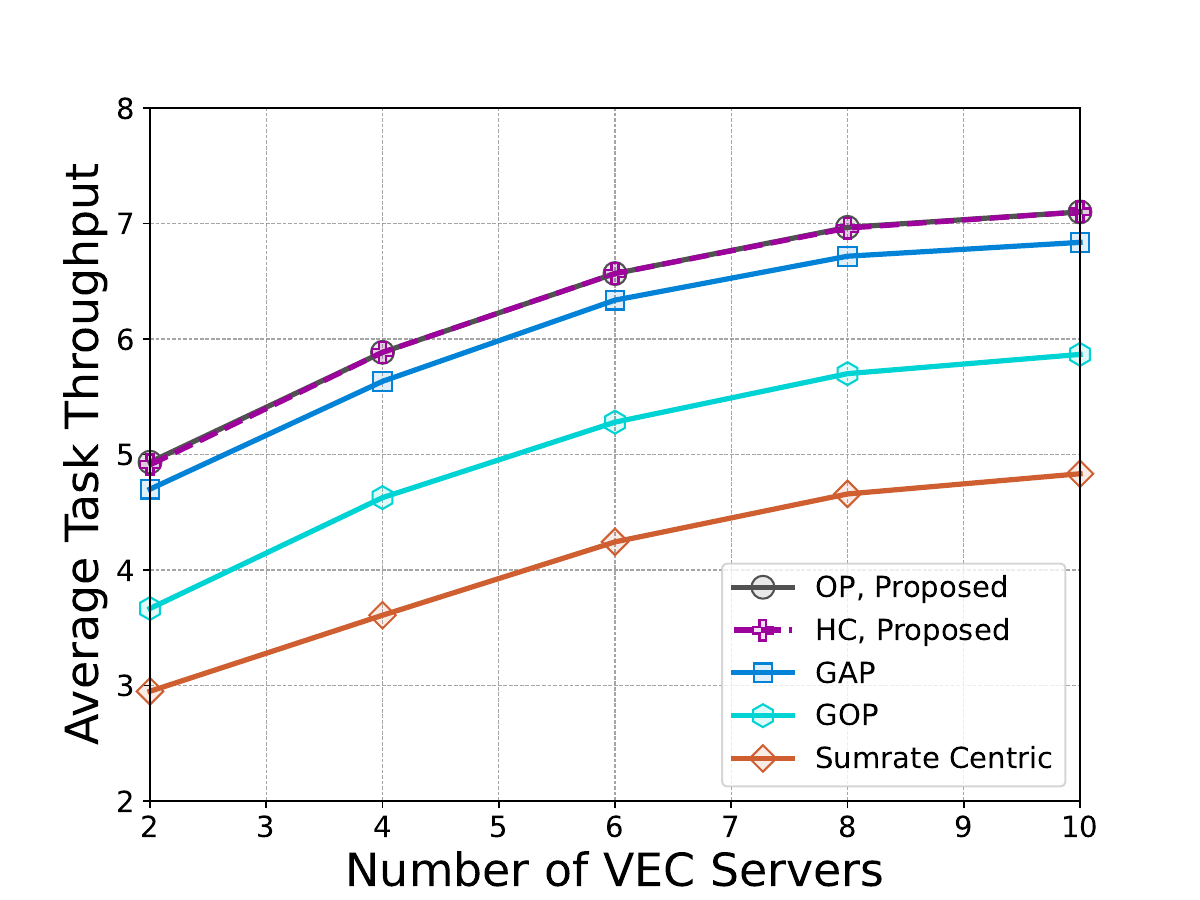}
    \caption{\rev{Vehicle arrival rate set to 0.5/s in highway environments}.}
    \label{fig:num_es_prob5}
\end{subfigure}
\caption{Average task throughput versus the number of VEC servers \rev{in urban roads (a)-(b) and highway environments (c)-(d)}.}
\label{fig:task_num_es}
\vspace{-0.2 cm}
\end{figure*}

\begin{figure*}[t]
    \centering
\begin{subfigure}[b]{0.24\textwidth}
    \includegraphics[width=\textwidth]{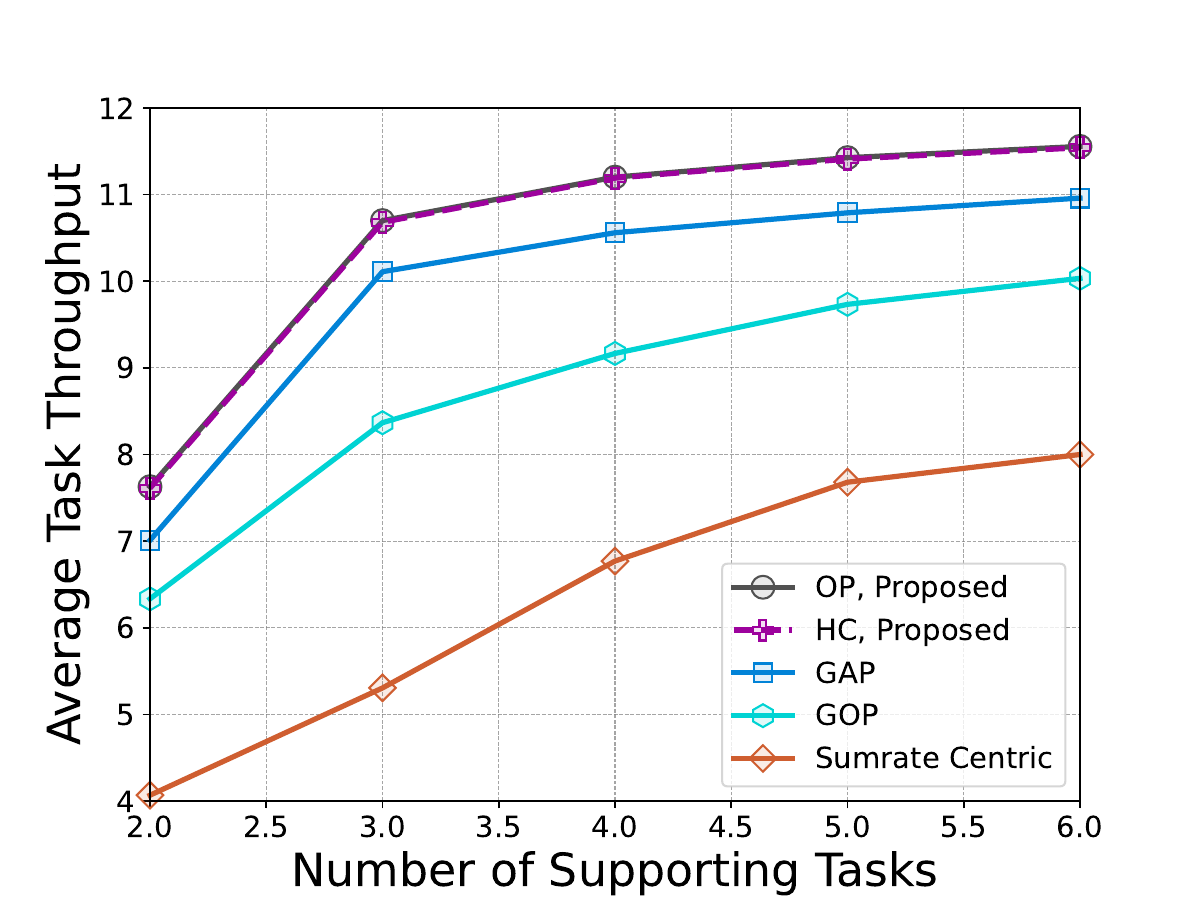}
    \caption{Vehicle arrival rate set to 0.7/s \rev{in urban roads}.}
    \label{fig:Cmax}
\end{subfigure}
    \hfill
\begin{subfigure}[b]{0.24\textwidth}
    \includegraphics[width=\textwidth]
    {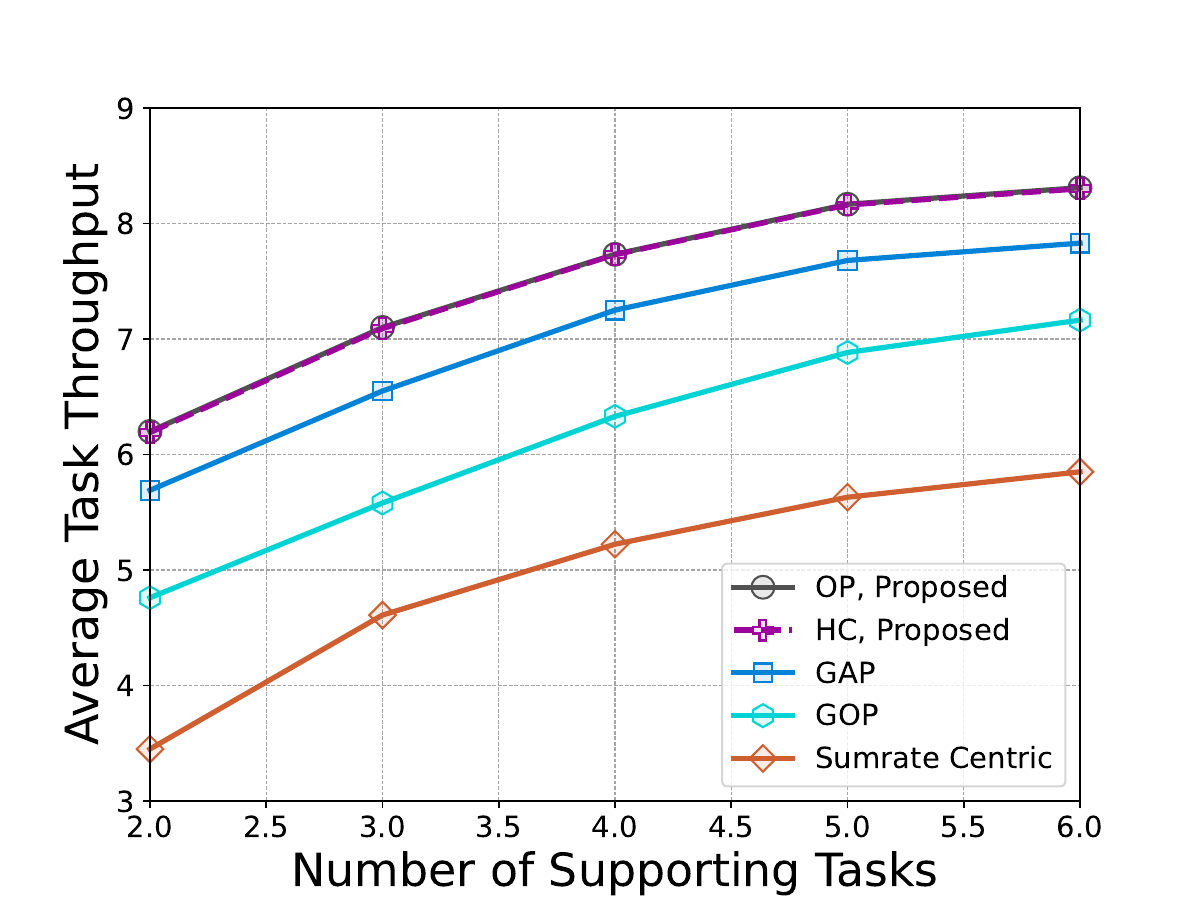}
    \caption{Vehicle arrival rate set to 0.5/s \rev{in urban roads}.}
    \label{fig:Cmax_prob5}
\end{subfigure}
    \hfill
\begin{subfigure}[b]{0.24\textwidth}
    \includegraphics[width=\textwidth]{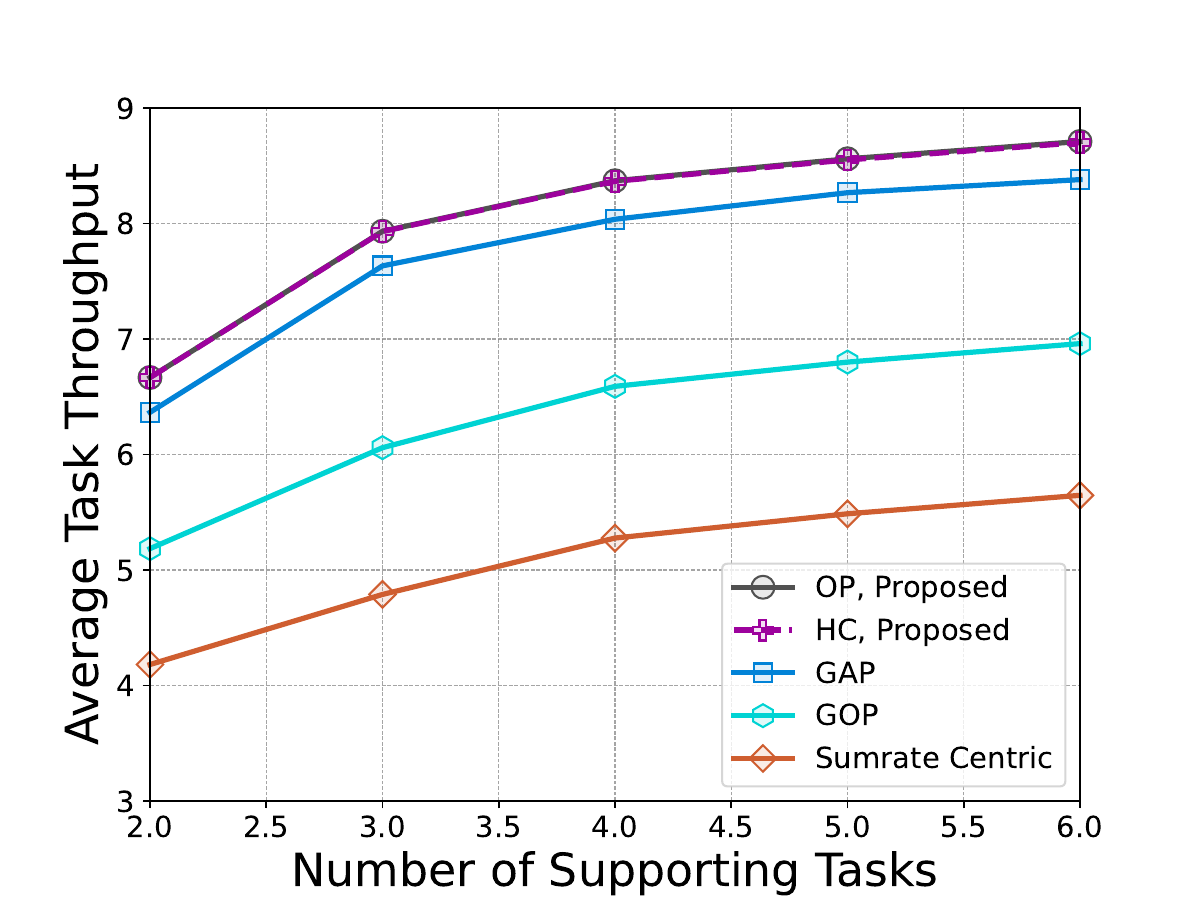}
    \caption{\rev{Vehicle arrival rate set to 0.7/s in highway environments}.}
    \label{fig:Cmax}
\end{subfigure}
    \hfill
\begin{subfigure}[b]{0.24\textwidth}
    \includegraphics[width=\textwidth]
    {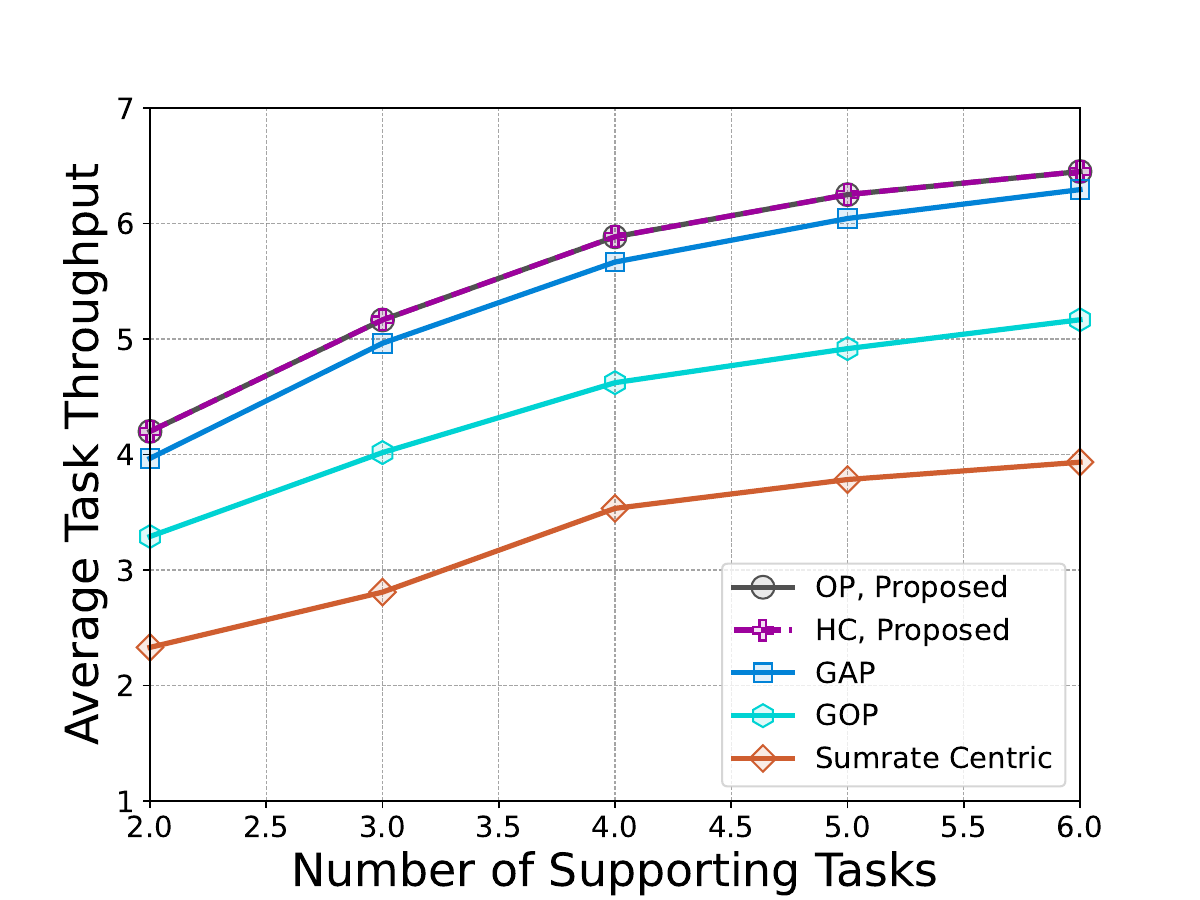}
    \caption{\rev{Vehicle arrival rate set to 0.5/s in highway environments}.}
    \label{fig:Cmax_prob5}
\end{subfigure}
\caption{Average task throughput versus the number of supporting tasks \rev{in urban roads (a)-(b) and highway environments (c)-(d)}.}
\label{fig:task_Cmax}
\vspace{-0.2 cm}
\end{figure*}

\begin{figure*}[t]
    \centering
\begin{subfigure}[b]{0.24\textwidth}
    \includegraphics[width=\textwidth]{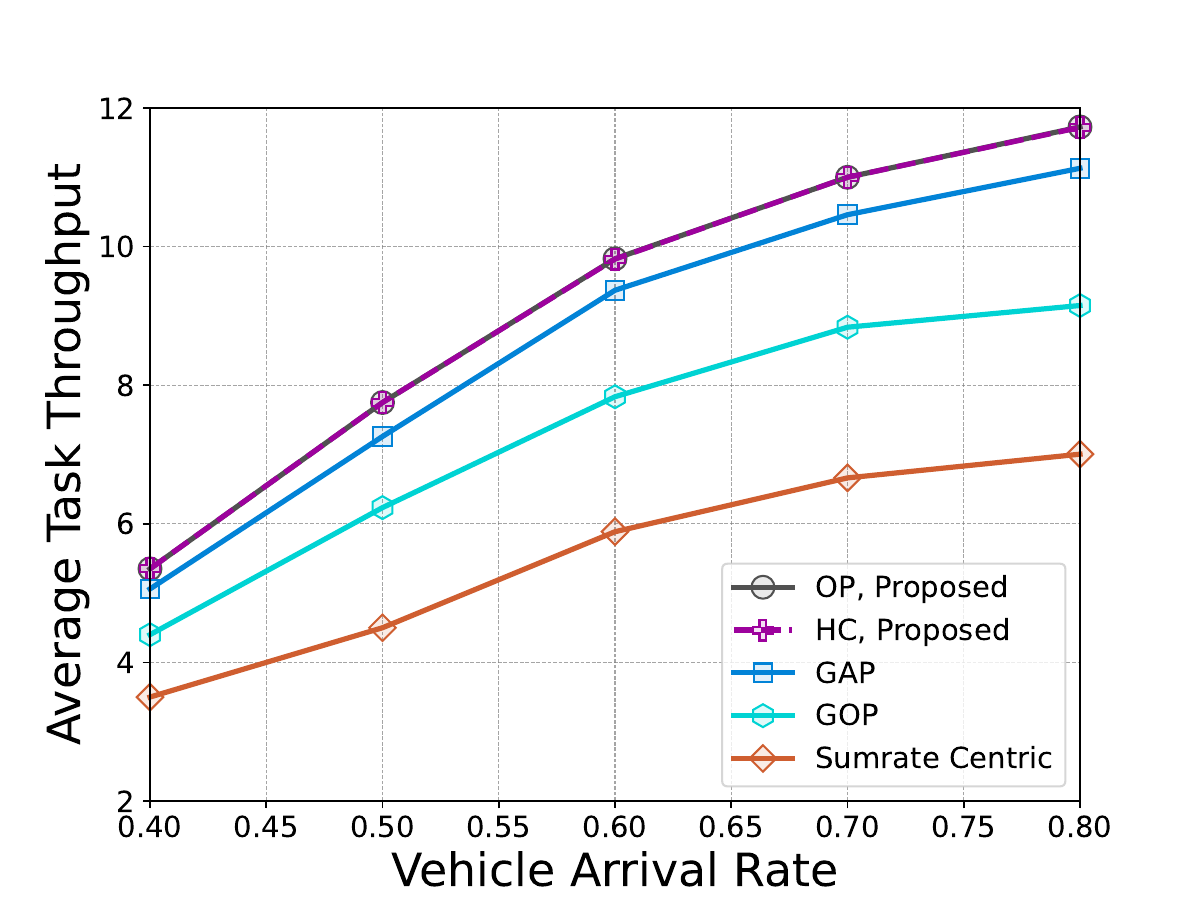}
    \caption{The number of VEC servers set to 4 \rev{in urban roads}.}
    \label{fig:density}
\end{subfigure}
    \hfill
\begin{subfigure}[b]{0.24\textwidth}
    \includegraphics[width=\textwidth]
    {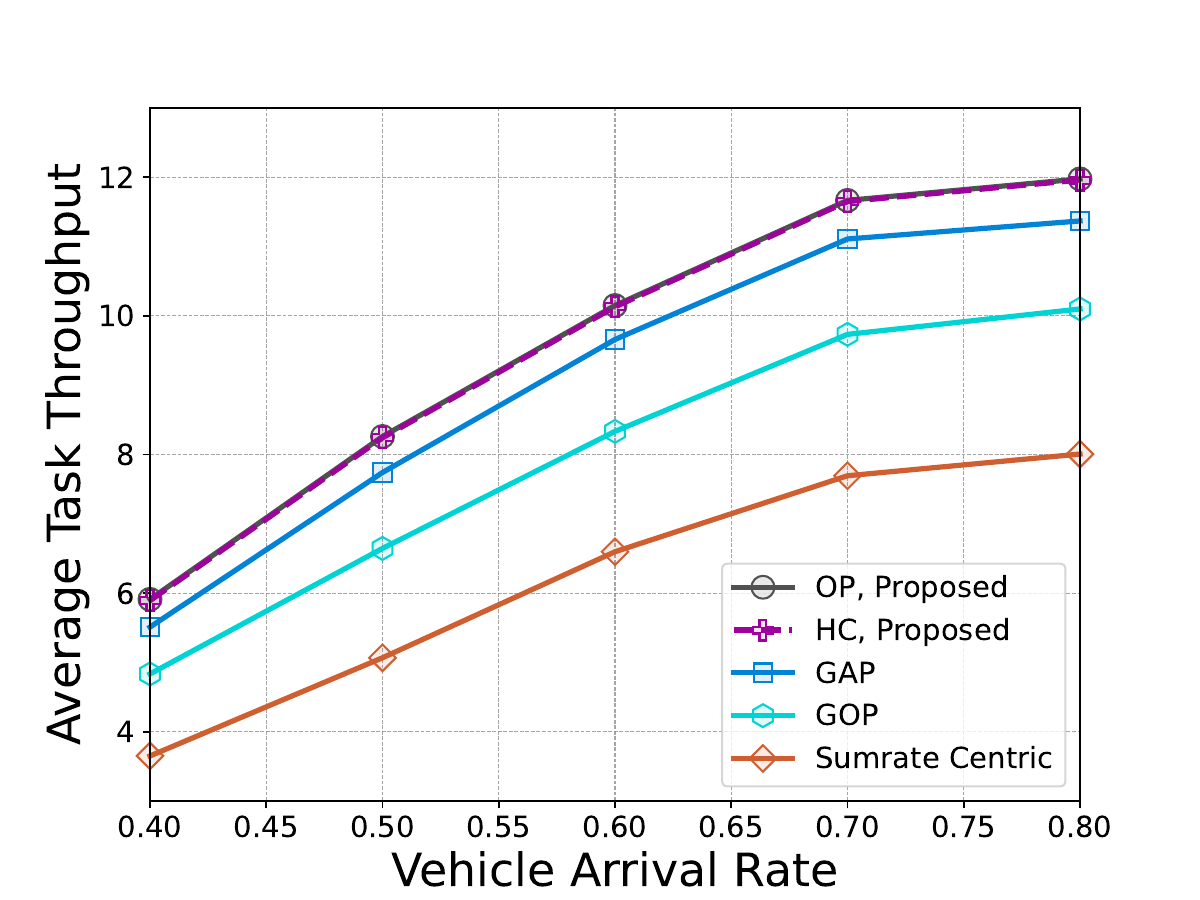}
    \caption{The number of VEC servers set to 6 \rev{in urban roads}.}
    \label{fig:density_prob5}
\end{subfigure}
    \hfill
\begin{subfigure}[b]{0.24\textwidth}
    \includegraphics[width=\textwidth]{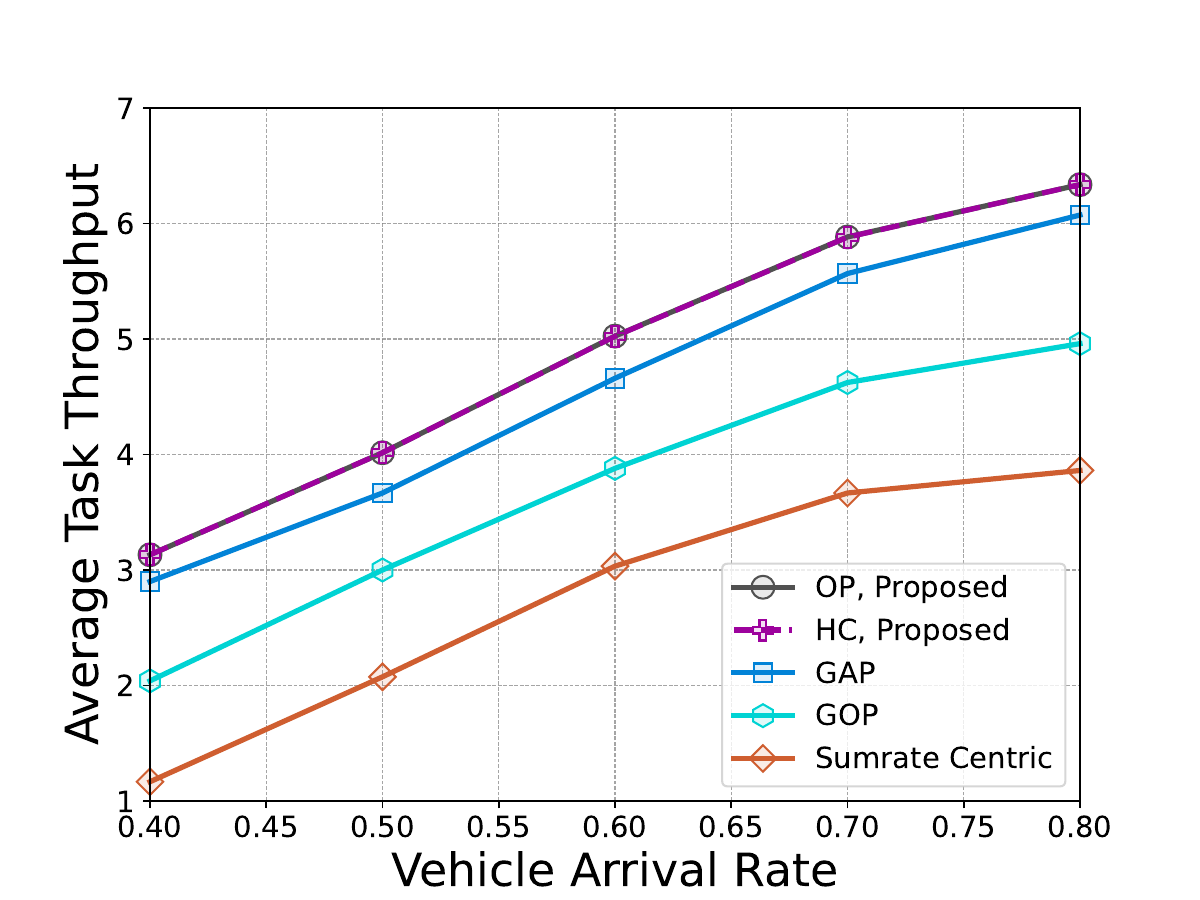}
    \caption{\rev{The number of VEC servers set to 4 in highway environments}.}
    \label{fig:density}
\end{subfigure}
    \hfill
\begin{subfigure}[b]{0.24\textwidth}
    \includegraphics[width=\textwidth]
    {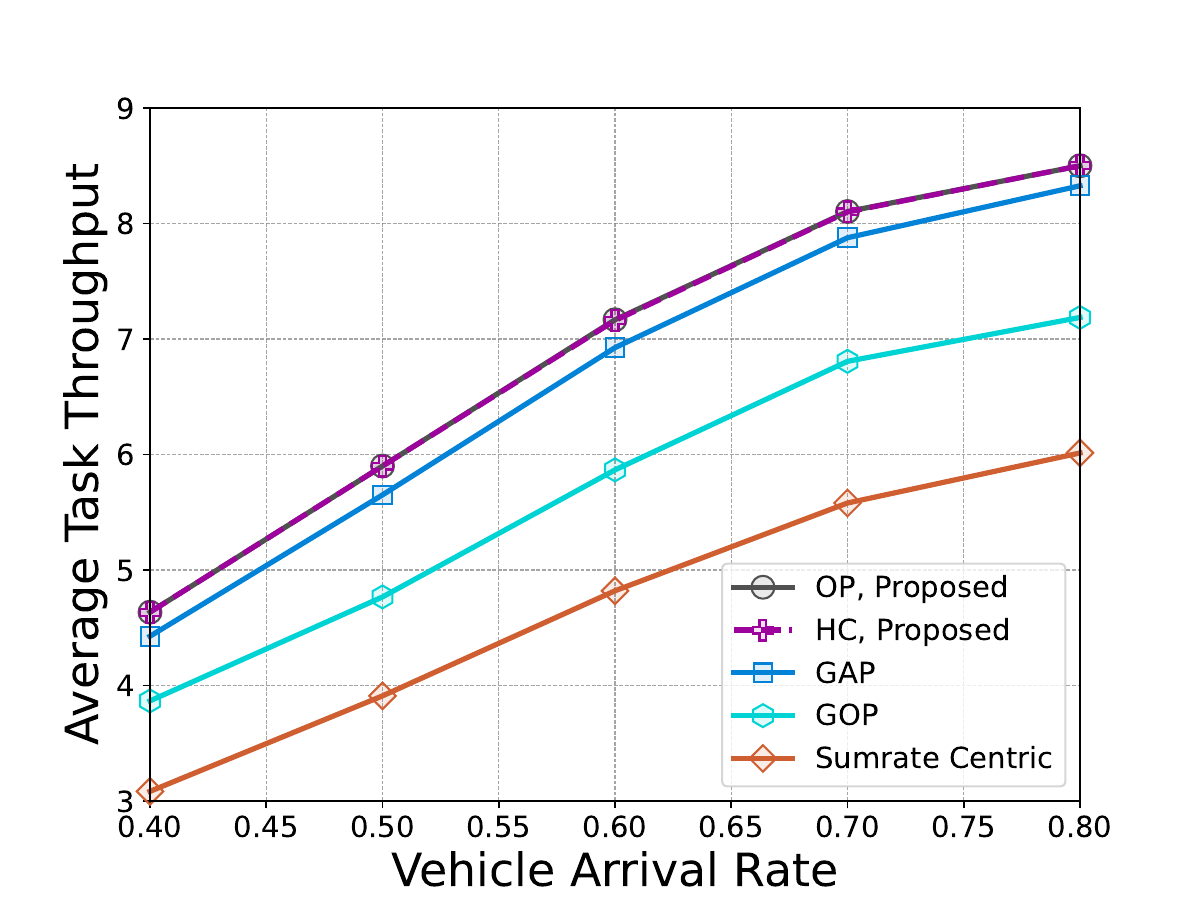}
    \caption{\rev{The number of VEC servers set to 6 in highway environments}.}
    \label{fig:density_prob5}
\end{subfigure}
\caption{Average task throughput versus the vehicle arrival rate \rev{in urban roads (a)-(b) and highway environments (c)-(d)}.}
\label{fig:task_density}
\vspace{-0.2 cm}
\end{figure*}

\begin{figure*}[t]
    \centering
\begin{subfigure}[b]{0.24\textwidth}
    \includegraphics[width=\textwidth]{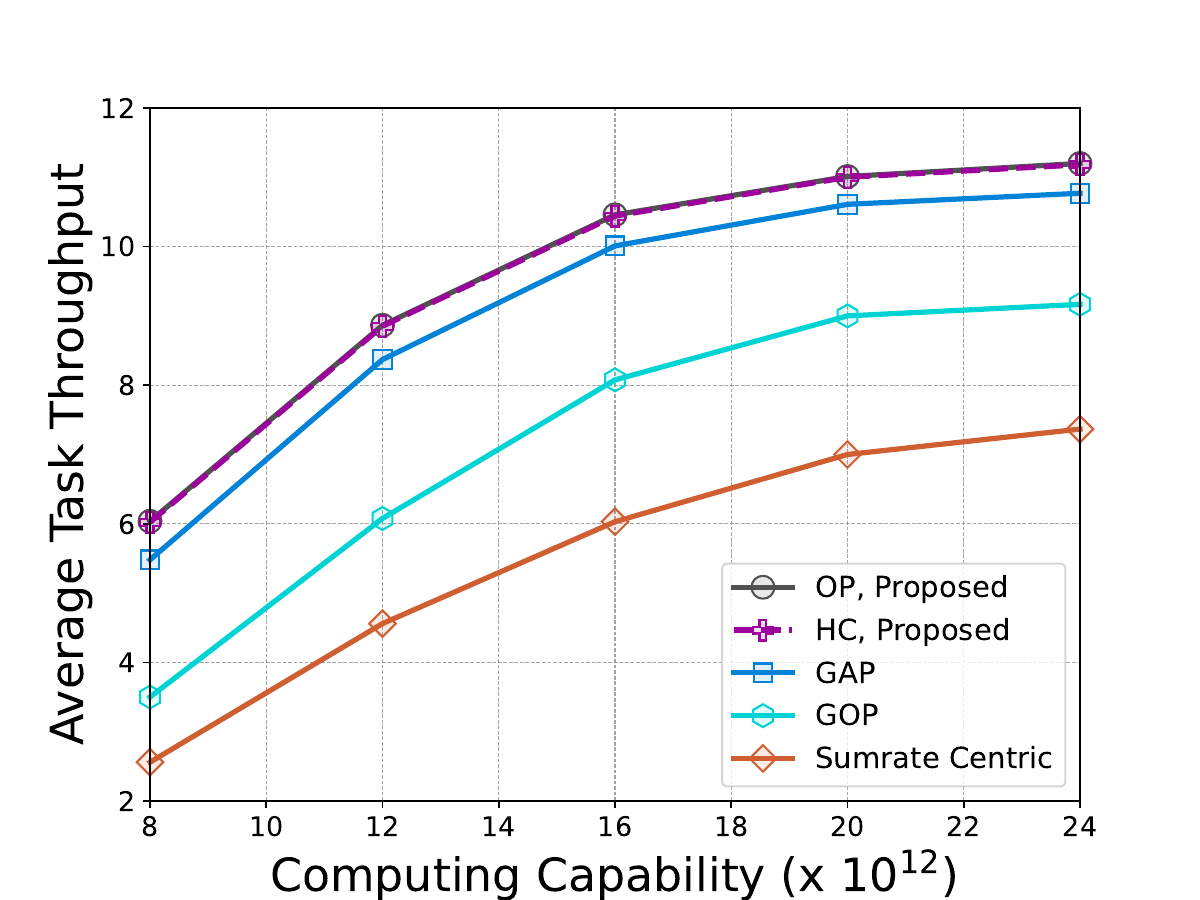}
    \caption{Vehicle arrival rate set to 0.7/s \rev{in urban roads}.}
    \label{fig:compt}
\end{subfigure}
    \hfill
\begin{subfigure}[b]{0.24\textwidth}
    \includegraphics[width=\textwidth]
    {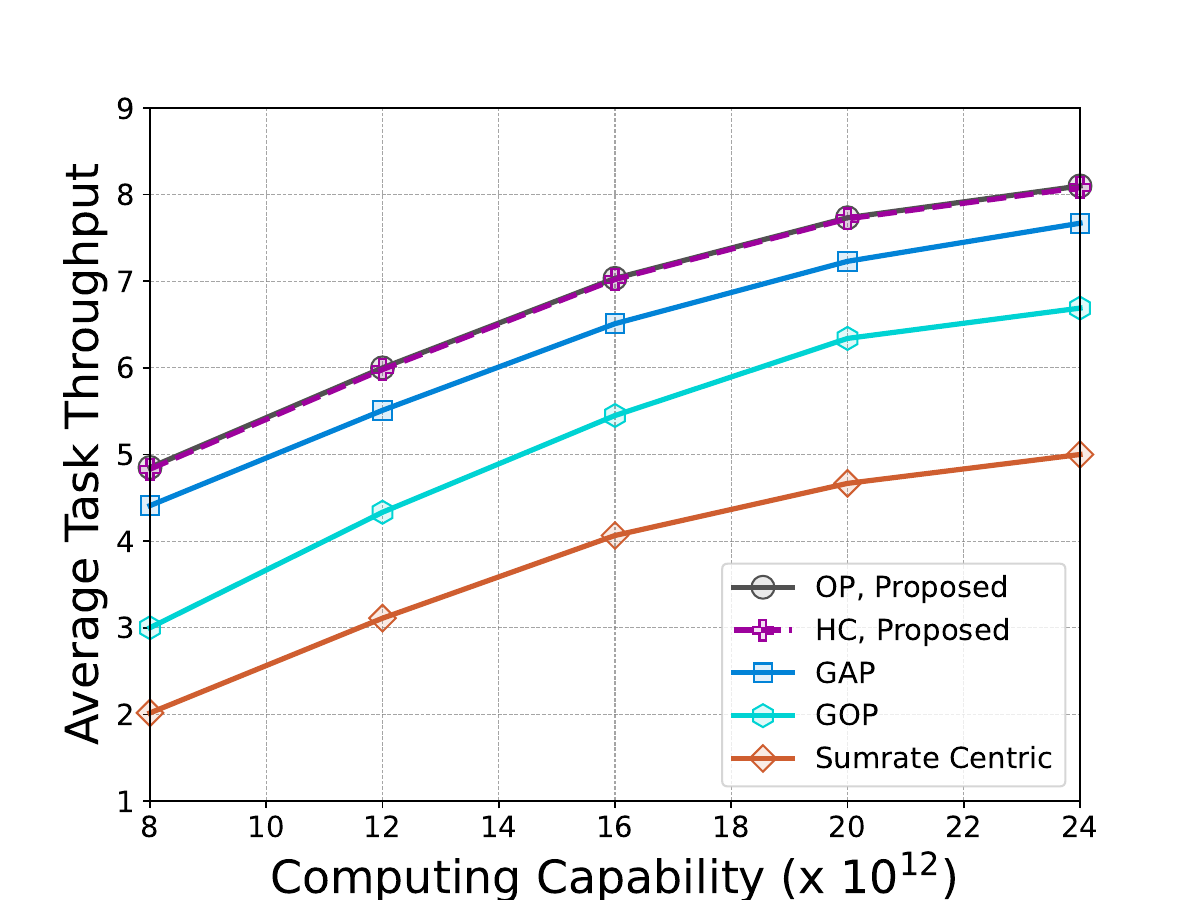}
    \caption{Vehicle arrival rate set to 0.5/s \rev{in urban roads}.}
    \label{fig:compt_prob5}
\end{subfigure}
    \hfill
\begin{subfigure}[b]{0.24\textwidth}
    \includegraphics[width=\textwidth]{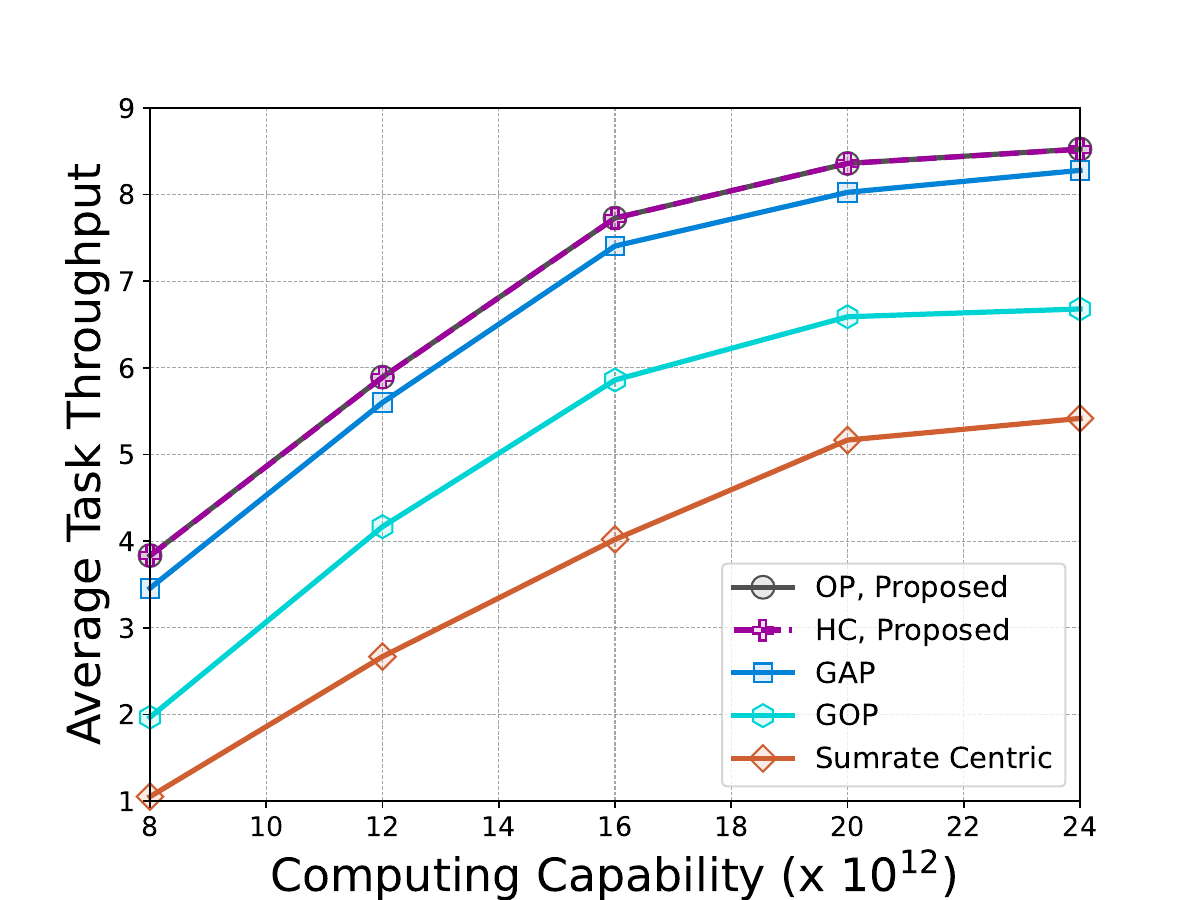}
    \caption{\rev{Vehicle arrival rate set to 0.7/s in highway environments}.}
    \label{fig:compt}
\end{subfigure}
    \hfill
\begin{subfigure}[b]{0.24\textwidth}
    \includegraphics[width=\textwidth]
    {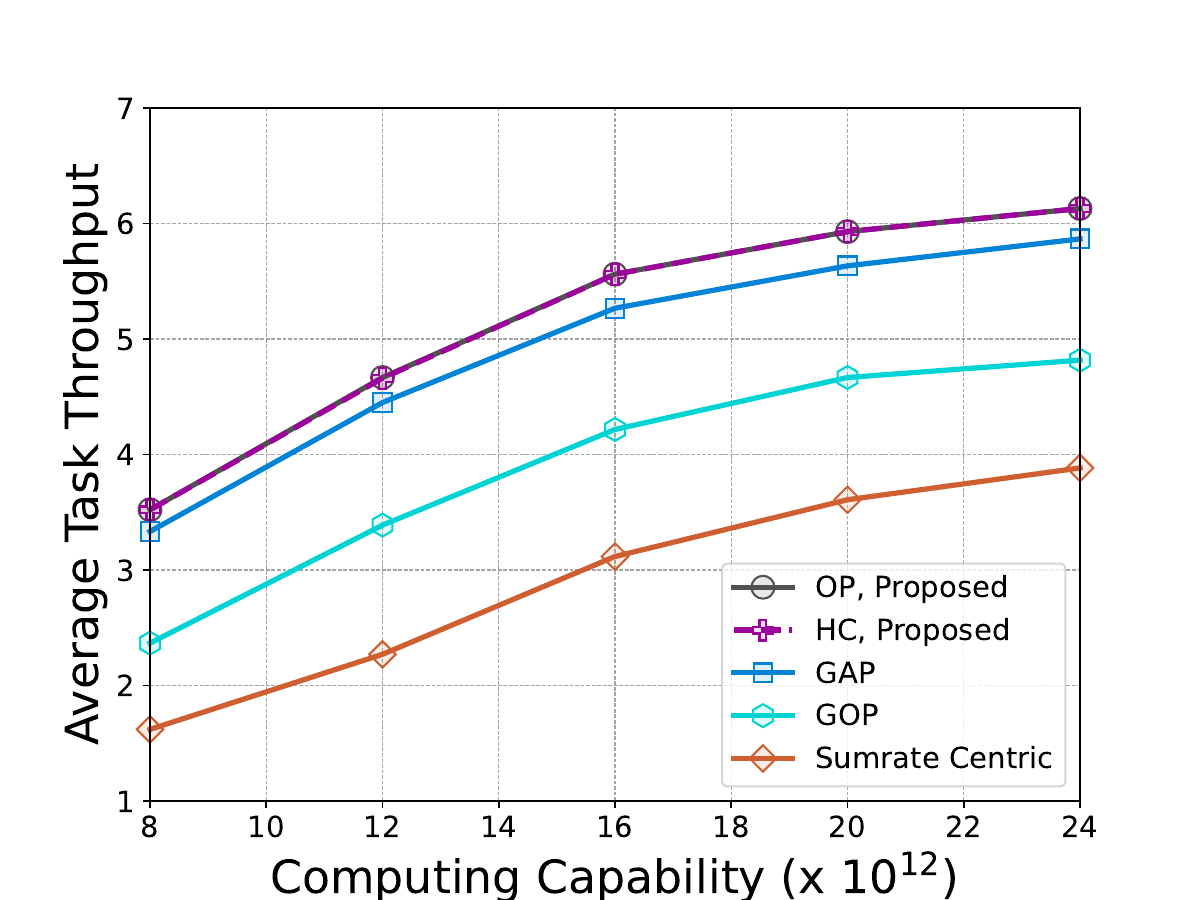}
    \caption{\rev{Vehicle arrival rate set to 0.5/s in highway environments}.}
    \label{fig:compt_prob5}
\end{subfigure}
\caption{Average task throughput versus the computing capability  \rev{in urban roads (a)-(b) and highway environments (c)-(d)}.}
\label{fig:task_compt}
\vspace{-0.2 cm}
\end{figure*}

\subsection{Numerical Results}
To evaluate the performance of our proposed algorithms, we consider the following benchmark methods:
\begin{itemize}
    \item Genetic-Algorithm-based Placement (GAP): In this approach, we consider a genetic algorithm-based heuristic for RIS placement based on the optimal task offloading obtained from Algorithm \ref{alg: Alg1}.
    \item Greedy-Offloading-based Placement (GOP): This algorithm offloads tasks by assigning each task to the nearest server while meeting the required task completion probability. If a server reaches its limit, this task is directed to the next available server. Based on the greedy task assignment, the RIS placement is optimized based on Algorithm \ref{alg: AlgOP} to maximize the average task throughput.
     \item \rev{Sumrate Centric \cite{Xiaowen}}: This approach optimizes the RIS placement to maximize the sum rate of all vehicles to all VEC servers. After determining RIS placement, the average task throughput is calculated accordingly. This scheme refers to traditional RIS placement to maximize system capacity instead of task throughput.
\end{itemize}

In Fig. \ref{fig:throughput_system}, we present the average task throughput in terms of the altitude and tilt angle of the RIS by varying the number of VEC servers. The visualized results exhibit a hill-like shape with one peak, serving as a justification for our HC-based algorithm. Specifically, the average task throughput initially increases with the RIS altitude and then decreases. This phenomenon occurs because a higher RIS improves the LoS probability and enables access to more servers, aligning with the principle of \textit{``standing higher, seeing farther"}. However, signal attenuation also increases considerably, negatively impacting task offloading performance. This underscores the importance of strategic RIS placement. Moreover, Fig. \ref{fig:throughput_system} also shows that as RIS altitude increases, the optimal tilt angle becomes steeper, tilting the RIS more toward the ground. This aligns with our intuition, as a higher RIS requires a sharper downward angle to maintain effective communication with ground-level vehicles. In addition, in Fig. \ref{fig:throughput_deployment_new} with four servers, the optimal RIS placement is (55 m, $69^\circ)$, while in Fig. \ref{fig:throughput_deployment_c3_new} with six servers, it is (62 m, $72^\circ)$. This finding suggests that the RIS should be deployed higher when there are more servers, as it can reach more available servers.

\begin{figure}[t]
    \centering
    \begin{subfigure}[b]{0.24\textwidth}
    \includegraphics[width=\textwidth]{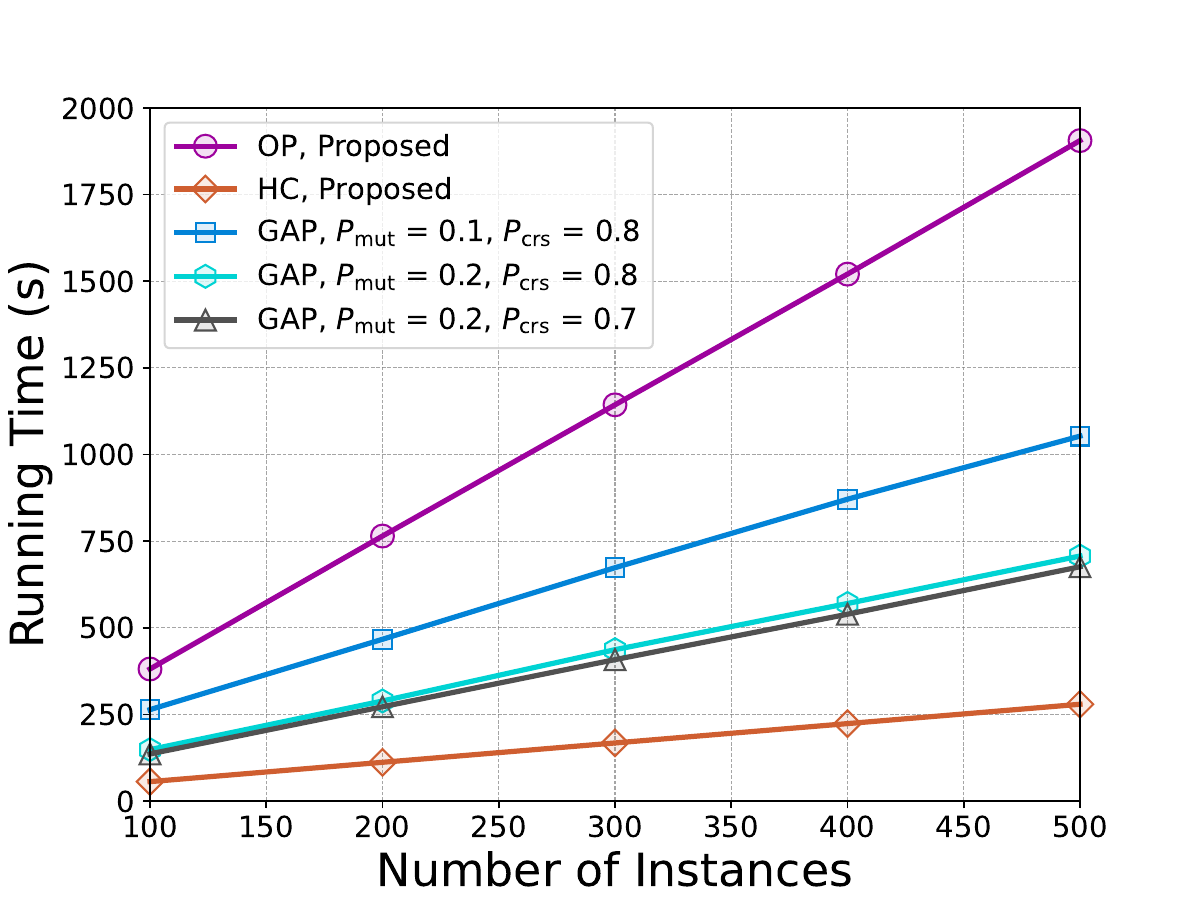}
    \caption{Running time versus the number of instances (the size of feasible set is 273).}
    \label{fig: time_ins}
\end{subfigure}
    \hfill
\begin{subfigure}[b]{0.24\textwidth}
    \includegraphics[width=\textwidth]{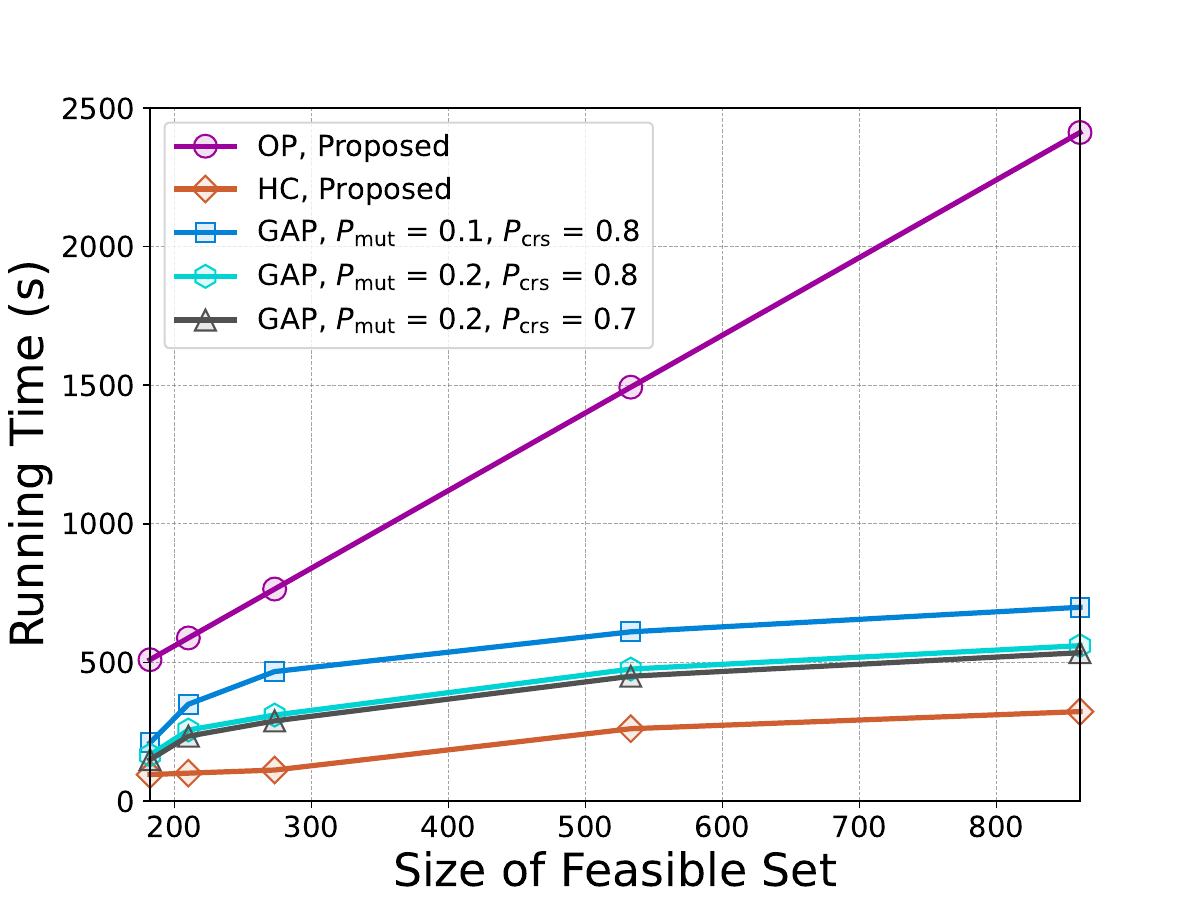}
    \caption{Running time versus the size of feasible set (the number of instances is 200).}
    \label{fig: time_size}
\end{subfigure}
    \caption{Running time of the proposed algorithms.}
    \label{fig: time}
    \vspace{-0.2 cm}
\end{figure}

Fig. \ref{fig:task_num_es} demonstrates the average task throughput versus the number of VEC servers, under different vehicle arrival rates. Impressively, the results show that the proposed HC algorithm achieves optimal performance in all of our simulation settings, i.e., matching the results found by optimal placement (OP). In comparison, the sumrate-centric algorithm performs the worst, as its RIS placement is optimized to maximize the sum rate from vehicles to servers rather than balancing both communication and computation loads with task deadline requirements. As the number of servers increases, average task throughput also rises accordingly, as more computing resources become available, which offers vehicles more opportunities to offload tasks. Notably, with 6 servers, the GOP algorithm struggles to efficiently manage the computing load since the task assignment is not optimal. This demonstrates the importance of jointly designing task offloading and RIS placement.

Fig. \ref{fig:task_Cmax} illustrates the average task throughput in relation to the number of supporting tasks per server, with results showing that our proposed framework consistently outperforms the benchmarks.
When the task capacity of individual servers is limited, effective task offloading within the VEC system becomes critical, and the proposed OP and HC algorithms demonstrate great advantages in these scenarios.
As the capacity for supporting tasks increases, average task throughput also improves, which is expected as each server can handle a greater number of computing tasks. Under these conditions, even the greedy assignment achieves relatively good performance.

The simulation results, presented in Fig. \ref{fig:task_density}, show the average task throughput relative to vehicle arrival rate across different numbers of servers. These results confirm that the proposed framework consistently achieves superior performance, particularly under high vehicle densities. By optimizing RIS placement, the proposed framework enables efficient connections between vehicles and potentially distant servers with lower computational loads, a balance that competing benchmarks struggle to achieve, especially with limited server availability.
Similarly, Fig. \ref{fig:task_compt} depicts the average task throughput versus server computing capability. A consistent trend emerges: as computing capability increases, so does the task throughput. However, once servers reach a certain level of computing power, system performance becomes constrained by communication limitations, reaching a steady state in task throughput.
\rev{The sumrate-centric benchmark consistently exhibits the worst performance, as its RIS placement strategy fails to balance communication and computation. These results underscore the importance of adopting throughput maximization as our metric in VEC systems.
}

\rev{To demonstrate the effectiveness of the proposed framework across various scenarios, we evaluate the performance in a challenging highway setting, where vehicle speeds range from 80 km/h to 130 km/h. As depicted in Fig. \ref{fig:task_num_es}-\ref{fig:task_compt}, our algorithm consistently outperforms alternative approaches in this scenario. These simulation results illustrate that the performance of VEC systems can be significantly improved through the joint optimization of RIS placement and task assignment, highlighting the robustness of our framework across various environments.
}

Finally, we demonstrate the efficiency of our proposed algorithms, as illustrated in Fig. \ref{fig: time}. Although the placement algorithms, i.e., the optimal placement and the hill climbing algorithms execute Algorithm \ref{alg: Alg1} in each step, implementing both algorithms is still practical because of the efficiency of Algorithm \ref{alg: Alg1}. In Fig. \ref{fig: time_ins}, the running time of the optimal placement algorithm increases linearly with both the number of instances and the feasible set. Other heuristic baselines are more scalable when the size of the feasible set increases since they can locate the solutions more swiftly. It is important to note that our proposed HC algorithm is much more efficient than both the OP approach and the GA-based algorithms, implying that our HC algorithm not only locates the optimal solutions in our experimental settings but also uses the shortest running time among the baselines.

\section{Conclusion}
\label{sec:conclusion}
In this paper, we have studied the RIS placement problem for multi-server vehicular edge computing (VEC) systems. Our goal is to strategically optimize the RIS placement, including its altitude and tilt angle, to balance the communication and computation workloads.
Considering the mobility of vehicles and wireless channels, we have developed a probabilistic channel model and a radiative gain model to characterize the vehicle-RIS-server channels as a function of the RIS positioning. 
To maximize average task throughput, we have developed efficient algorithms for optimizing RIS placement, including a grid search procedure and a hill-climbing (HC) algorithm, with each step executing our optimal task offloading optimization. Interestingly, we have discovered that the hill climbing algorithm can often find its optimal solution. We have conducted extensive experiments and demonstrated that our proposed RAISE framework significantly outperforms benchmark methods. We expect that our work offers an effective solution to harnessing more computing power and significantly boosting task throughput for future VEC systems.

 




\vfill

\end{document}